\documentclass[12pt, reqno]{amsart}
\usepackage{amsmath,amssymb,amsfonts,amscd,hyperref,color}
\usepackage[utf8]{inputenc}
\usepackage{verbatim}
\usepackage{mathtools}
\usepackage{tikz}
\usetikzlibrary{matrix,calc}
\usetikzlibrary{positioning}
\usetikzlibrary{matrix}
\usepackage{pgfplots}
\usetikzlibrary{patterns}

\newcommand{\Hmm}[1]{\leavevmode{\marginpar{\tiny%
$\hbox to 0mm{\hspace*{-0.5mm}$\leftarrow$\hss}%
\vcenter{\vrule depth 0.1mm height 0.1mm width \the\marginparwidth}%
\hbox to
0mm{\hss$\rightarrow$\hspace*{-0.5mm}}$\\\relax\raggedright #1}}}

\newtheorem{theorem}{Theorem}[section]

\newtheorem{lemma}[theorem]{Lemma}
\newtheorem{proposition}[theorem]{Proposition}

\theoremstyle{definition}

\newtheorem{assumption}[theorem]{Assumption}
\newtheorem*{remark}{Remark}
\newtheorem*{convention}{Convention}

\numberwithin{equation}{section}
\newcommand{\Z}{{\mathbb Z}}
\newcommand{\R}{{\mathbb R}}
\newcommand{\C}{{\mathbb C}}

\newcommand{\N}{{\mathbb N}}

\newcommand{\lorenzorrection}[1]{\#1}%{\textcolor{red}{#1}}}

\DeclarePairedDelimiter\floor{\lfloor}{\rfloor}

\begin{document}

\date{\today}

\author[M. Keller]{Matthias Keller}
\address{M. Keller: Israel Institute of Advanced Studies (IIAS), The Hebrew University of Jerusalem, Edmond J. Safra Campus, Givat Ram, Jerusalem, Israel}
\address{Institut f\"ur Mathematik, Universit\"at Potsdam, 14476 Potsdam, Germany}
\email{matthias.keller@uni-potsdam.de}

\author[L. Pettinari]{Lorenzo Pettinari}
\address{L. Pettinari: Dipartimento di Matematica, Universit\`a di Trento and INFN-TIFPA and INdAM, Via Sommarive 14, I-38123 Povo, Italy}
\email{lorenzo.pettinari@unitn.it}

\author[C. J. F. van de Ven]{Christiaan J. F. van de Ven}
\address{C. J. F. van de Ven: Friedrich-Alexander-Universit\"at Erlangen-N\"urnberg, Department of Mathematics, Cauerstra\ss e 11, 91058 Erlangen, Germany}
\email{chris.ven@fau.de}

\title[Convergence of eigenvalues]{Coupling of the continuum and semiclassical limit. Part I: convergence of eigenvalues}

\begin{abstract}
\noindent
We analyze the semiclassical  $d$-dimensional Schr\"{o}\-ding\-er operator in the continuum $ \frac{1}{2} \Delta + \lambda_N^2 V$
discretized on a mesh with spacing proportional to $1/N$. The semi-classical parameter $\lambda_N$ is chosen as $\lambda_N = N^{1 - \gamma}$, with $\gamma \in (-1,1)$, which ensures that $N$ governs both the semiclassical and continuum limit simultaneously. We prove that all eigenvalues of the discrete operator converge to those of the continuum, as $\lambda_N\to\infty$. Beyond this semi-classical domain, in the case of the harmonic oscillator, we further discuss the spectral asymptotics for $\gamma \in \mathbb{R} \setminus (-1,1)$, thereby fully characterizing the eigenvalue behavior across all possible values of $\gamma\in\mathbb{R}$.
\end{abstract}

\maketitle

\tableofcontents
\section{Introduction}
This is the first part of a work which is devoted
to the analysis of the spectral asymptotics of a discrete Schr\"{o}dinger operator. Subsequently there is a second part \cite{Keller_Pettinari_Ven_Part2} which will then be devoted to the analysis of the eigenfunction asymptotics in the same limit, including Agmon-type estimates for the ground state.

This first part is focused on the convergence of eigenvalue to the analysis of the energy asymptotics of a $d$-dimensional discrete Schr\"{o}dinger operator in a \emph{combined} semi-classical and continuum limit. Unlike other approaches in this context \cite{Isozaki_Jensen_2022,Kameoka_2023,Nakamura_2021,Klein_Rosenberger_2008,Klein_Rosenberger_2009}, we investigate a wide range of scalings, interlacing the semiclassical parameter with the discretization one. This coupling yields a rich variety of asymptotic regimes for the eigenvalues. In particular, we are able to pinpoint the precise region in parameter space where the limiting behavior aligns with the semiclassical limit of a continuum Schr\"{o}dinger operator $ \frac{1}{2} \Delta_{\mathbb{R}^d} + \lambda_N^2 V$ in $d$ dimensions.
More precisely, we rigorously show that, in the semi-classical regime $\lambda_N \gg 1$, with $\lambda_N = N^{1 - \gamma}$ for $\gamma \in (-1,1)$ the scaling parameter, all eigenvalues of the discrete Schr\"{o}dinger operator exhibit the same asymptotic behavior as those of the continuum operator $ \frac{1}{2} \Delta_{\mathbb{R}^d} + \lambda_N^2 V,$
where the smooth, non-negative potential $V$ with finitely many non-degenerate minima is strictly positive at infinity. This result confirms the validity and appropriateness of our scaling relation between the semi-classical parameter $\lambda_N$ and the discretization mesh size $1/N$, showing that it effectively captures the classical limit of the continuous model through the discrete approximation. As already mentioned above analysis of Agmon-type estimates for eigenfunctions are deferred to our forthcoming companion paper \cite{Keller_Pettinari_Ven_Part2}.

In the final part of the paper, we emphasize that our analysis on eigenvalue asymptotics extends beyond the semi-classical regime. It not only identifies the interval $(-1,1)$ as the true semiclassical domain, but also classifies four additional regions where the spectral behavior deviates significantly. Framed within this broader context, our study provides a unified understanding of the various limiting regimes and underscores the role of the parameter $\gamma$ in determining the spectral asymptotics.

\subsection{Semiclassical analysis and discrete Schr\"{o}dinger operators}

For a positive integer $d\geq 1$, we consider the lattice $\Z^d$.  Let $V:\R^d\to \R$ be a function  and  define
\begin{align*}%\label{pot}
V_N:\Z^d\to \R,\quad  x\mapsto V\left(\frac{x}{N}\right).
\end{align*}
For a given parameter $\gamma\in \mathbb{R}$ and $N\in\N$, we consider the following discrete Schr\"{o}dinger operator on $\ell^2(\Z^d)$ 
\begin{align*}%\label{iyi}
H_Nf(x):=-\frac{N^2}{2}\sum_{|x-y|=1}(f(y)-f(x))+N^{2(1-\gamma)}V_N(x)f(x)
\end{align*}
For suitable conditions on $V$, \textit{cf.} Assumption \ref{assumption} in Section \ref{subsec:main}, we analyze the spectral asymptotics as $N\to\infty$ for different values of $\gamma$. 
%In short,  In Section \ref{sec: Alternative regions} we provide an overview of the eigenvalue asymptotics of the other regimes of the parameter $\gamma$, starting from $H_N$.
In particular, we may identify a total number of five distinct limits with the following characteristics for the eigenvalues, as further discussed in Remark \ref{rmk:wichtig}, namely
\begin{itemize}
    \item $\gamma> 1$: corresponding to the continuum limit of the free discrete Laplacian, towards the Laplacian $-\frac{1}{2}\Delta_{\mathbb{R}^d}$;
    \item $\gamma=1$: corresponding to the continuum limit of $H_N$ towards the operator $ \frac{1}{2}\Delta_{\mathbb{R}^d}+V$;
    \item $\gamma\in (-1,1)$: corresponding to a semi-classical limit of the operator  $\frac{1}{2}\Delta_{\mathbb{R}^d}+\lambda^2 V$ on $\mathbb{R}^d$, with main result stated in Theorem~\ref{thm:main} below; 
    \item $\gamma=-1$: corresponding to a purely discrete model, whose precise features depend on the particular choice of the potential;
    \item $\gamma<-1$: corresponding to a semiclassical approximation of a discrete model.
\end{itemize}
Among these regimes, the case $\gamma\in(-1,1)$ is of particular interest: it yields not only semiclassical spectral asymptotics, but also allows for the derivation of exact Agmon estimates for the ground state, in the spirit of the continuum results of \cite{Simon_1984}. This regime constitutes the main focus of our second work \cite{Keller_Pettinari_Ven_Part2}.

%\textcolor{red}{
\subsection{Physical interpretation and comparison with literature}
We recall the usual semi-classical analysis for a Schr\"{o}dinger operator in the continuum
\begin{align*}%\label{cts Schroedinger 1}
   \frac{h^2}{2}\Delta_{\R^d}+V \ \ \ \text{on} \ \ L^2(\mathbb{R}^d),
\end{align*}
where $\Delta_{\R^d}=-\sum_{i=1}^d\partial_i^2$ is the usual Laplacian on $\mathbb{R}^d$, and $V:\R^d\to\R$ is a potential defined as multiplication operator.
The parameter $h$ is a scaled and dimensionless version of Planck's constant, describing a classical theory in the regime $h\to 0$. The behavior of this operator is closely related to the behavior of the rescaled operator with $\lambda=1/h$
\begin{align}\label{cts Schroedinger 1}
   \frac{1}{2}\Delta_{\R^d}+ {\lambda^2} V \ \ \ \text{on} \ \ L^2(\mathbb{R}^d),
\end{align}
when letting $\lambda\to\infty$. This limit is studied by Simon in \cite{Simon_1984,Simon_1985} in relation to  the time dependent Schr\"{o}dinger equation.
\newline
For a discrete analogue, one typically considers  a uniform mesh size $\delta$ and introduce a discrete Schr\"{o}dinger operator $H_{\delta,\lambda}$ on $\ell^2(\delta\mathbb{Z})$, defined by
 \begin{align*} 
     H_{\delta,\lambda}f(x)=-\frac{1}{2\delta^2}\sum_{|x-y|=\delta}(f(y)-f(x))+\lambda^2V(x)f(x).
 \end{align*}
The continuum limit deals with the analysis of this operator as $\delta\to 0$ at fixed $\lambda$. Details on this limit can be found in e.g. \cite{Isozaki_Jensen_2022,Nakamura_2021}. In contrast, the classical limit concerns the regime $\lambda\to \infty$, at fixed $\delta$ (see e.g. \cite{Bach_Pedra_Lakaev_2017}).
\newline
In this work, we couple the parameters $\delta,\lambda>0$ as follows. For a given parameter $\gamma\in \mathbb{R}$ and $N\in\N$, we let $$\delta_N=1/N\quad\mbox{ and }\quad\lambda_N=N^{1-\gamma}.$$
In view of the above discussion, the corresponding  operator is  given by
\begin{align*} 
   {  H_{1/N,N^{1-\gamma}}}f(x)=-\frac{N^2}{2}\sum_{|x-y|=1/N}(f(y)-f(x))+\lambda_N^2V(x)f(x)  
 \end{align*}
{on}  $\ell^2(\mathbb{Z}^d/N).$ Thus, the limit $\delta\to 0$ is {not} taken independently of $\lambda_N$; instead, it is a  single coupled limit $N\to\infty$ obeying $\delta_N=1/N$ and $\lambda_N=N^{1-\gamma}$. In other words, the limit $N\to\infty$  {\em simultaneously} describes the continuum as well as the semi-classical limit.

To see how $H_N$ defined above is derived from $  {  H_{1/N,N^{1-\gamma}}}$ we make the following simplification, which, however, does not lead to a loss of generality. Instead of considering the scaled lattice $\Z^d/N$, we work on $\Z^d$, which means that we have to consider the potential $V_N:\Z^d\to \R$ given by 
$$V_N(x) =V\left(\frac{x}{N}\right),$$
and furthermore the discrete Laplacian on $\Z^d$ sums over vertices of distance $1$ instead of distance $1/N$, which readily gives the operator $H_N$ on  $\ell^2(\mathbb{Z}^d )$ which we already introduced in the beginning of this section, i.e.,
\begin{align*}\label{generic hamiltonian}
     H_{N}f(x)=-\frac{N^2}{2}\sum_{|x-y|= 1}(f(y)-f(x))+N^{2(1-\gamma)}V_N(x)f(x).
 \end{align*}
Note that $H_N$ unitarily equivalent to $   {  H_{1/N,N^{1-\gamma}}}$ via the map $U_N:\ell^2(\Z^d)\to \ell^2(\Z^d/N)$, $U_Nf(x/N)=f(x)$. In particular, the prefactor $N^2$ in front of the Laplacian is essential: it represents the approximate second derivative of a function in the continuum.
\begin{remark}\label{rmk:wichtig}
We compare our results with the current literature, taking  $H_{N}$ as the starting point. \smallskip

$  \gamma<-1$: By rescaling the operator with the global factor $N^{-2}$, one obtains, for $\gamma<-1$, a semiclassical approximation of the discrete model. This  limit was analyzed in the sense of Weyl asymptotics  for bounded multiplication operators in \cite{Bach_Pedra_Lakaev_2017}, where {asymp}\-totic formulas for the number of discrete eigenvalues in the semi-classical limit  were obtained. It has been studied in \cite{Linn_Lippner_Yau_2012} in the case of  finite connected graphs. Regarding the latter work, in case of the harmonic oscillator, we obtain similar results, in the sense that the eigenvalues converge towards the points of the potential, discussed in \textit{cf.} Section \ref{sec: Alternative regions}. In addition, the case $\gamma=-1$ corresponds to a purely discrete model, whose eigenvalue asymptotics coincide with those of the discrete harmonic oscillator, \textit{cf.} Proposition \ref{prop: -inf to -1}.\smallskip

$ \gamma\in (-1,1)$: For the specific choice of $\gamma=0$ where the operator $H_N$ approximates the semi-classical limit of \eqref{cts Schroedinger 1}, various interesting results  have been obtained in case of bounded potentials with sufficient regularity. Examples include the works by Helffer and Sj\"{o}strand \cite{HorHel_1988}, Kameoka \cite{Kameoka_2023}, Rabinovich \cite{Rabi}, and Klein and Rosenberger \cite{Klein_Rosenberger_2009}. In particular, Klein and Rosenberger studied a generalized Laplacian with variable coefficients, which in our setting reduces to the constant-coefficient Laplacian.
    Unlike the present work, the spectral estimates for the eigenvalues in these references are obtained using microlocal techniques, and eigenvector asymptotics are derived in terms of the Agmon–Finsler metric \cite{Kameoka_2023,Klein_Rosenberger_2008}.  A further extension of these asymptotics to eigenfunctions will be presented in our follow-up work \cite{Keller_Pettinari_Ven_Part2}. To the best of our knowledge, there is no literature on this regime for $\gamma\in (-1,1)\setminus\{0\}$, and the case of unbounded potentials has not been studied neither.  \smallskip

$\gamma=1$: For $\gamma=1$, i.e., the regime in which the continuum limit of $H_N$ approximates $\tfrac{1}{2}\Delta + V$, convergence results have been obtained in the norm-resolvent sense. In particular, Nakamura and Tadano \cite{Nakamura_2021} prove norm-resolvent convergence under suitable assumptions on the potential, and Kameoka and Nakamura \cite{Kameoka_2024} further extend this analysis to the convergence of resonances. In the present work we only discuss this regime for the harmonic oscillator as particular case of Theorem~\ref{thm:eigenvalue_asymptotics_harmonic_oscillator_main}.\smallskip

$ \gamma>1$: Since the eigenvalues tend to zero as $O(\lambda_N)$, we expect the limit to be the Laplacian with purely continuous spectrum, as follows from Theorem \ref{thm:eigenvalue_asymptotics_harmonic_oscillator_main}. 

\end{remark}
%}

\subsection{Main results on eigenvalue asymptotics}\label{subsec:main}

We are interested in the asymptotic behavior of the eigenvalues of $H_N$ as $N\to\infty$ for potentials which satisfy the following assumption.

\begin{assumption}\label{assumption}
The following conditions on $V:\R^d\to \R$ are required:
    \begin{enumerate}
	\item[(V1)] $V\in C^{\infty}(\mathbb{R}^d)$ and $V\geq 0$.
	\item[(V2)] $V$ has $1\leq m <+\infty$ zeros, $a_1,\dots a_m$, such that the Hessian at these points  given by    
    \begin{equation*}%\label{eq: Hessians}
	   D^2V(a_l) = ((\partial_{i}\partial_{j}V)(a_l))_{i,j=1}^d
	\end{equation*}
    is strictly positive definite for all $l=1,\ldots,m$, i.e., $a_1,\ldots, a_m$ are non-degenerate minima for $V$.
	\item[(V3)] $V$ is strictly positive at infinity, i.e., there exists a $c >0 $ and $R_{0}>0$ such that $|x|> R_{0}$ implies $V(x)\geq c$.
    \end{enumerate}
\end{assumption}
Due to Assumption~\ref{assumption}~(V1), the Hessian of $V$ is symmetric and, due to (V2), it has $d$ strictly positive eigenvalues at each zero $a_1,\ldots,a_m$ which we denote by $\omega_i(a_l)$ for $i=1,\ldots,d$, $l=1,\ldots,m$.     Furthermore, the set
\begin{align*}
    \Sigma(V)&=\left\{\frac12\sum_{i=1}^d \omega_i(a_l)(2n_i+1) \mid n_i\in \N_0, l=1,\ldots,m\right\}
   \end{align*}
   can be enumerated as a non-decreasing sequence of eigenvalues counting multiplicities, i.e., 
\begin{align*}
    e_0(V)\leq e_1(V)\leq\ldots \leq e_n(V)\leq\ldots
\end{align*}
and it is known that  the energy levels $e_n(V)$ are precisely those obtained from a semi-classical limit starting with $\Delta_{\R^d}+\lambda_N^2 V$, with $\lambda_N=N^{1-\gamma}$ and $\gamma\in (-1,1)$, \textit{cf.} \cite[Theorem~1.1]{Simon_1983}. 

We consider the $n$-th eigenvalue of the operator $H_N$ counting with multiplicity which we denote by $E_n(H_N)$ and show the following eigenvalue asymptotics.

\begin{theorem}\label{thm:main}
    Let $V$ be a potential satisfying Assumption~\ref{assumption} and $\Sigma(V)=\{e_0(V)\leq e_1(V)\leq\ldots\leq e_n(V)\leq \ldots\}$ as above.  Then, for all $n\in \mathbb{N}_0$ and $\gamma\in (-1,1)$, it holds
    \begin{align*}
        \lim_{N\to \infty} \frac{E_n(H_N)}{\lambda_N} = e_n(V),
    \end{align*}
    where $\lambda_N=N^{1-\gamma}$.
\end{theorem}

The paper is structured as follows. In Section~\ref{sec: Discrete harmonic oscillators} we study the discrete harmonic oscillator and prove eigenvalue asymptotics in this special case. In Section~\ref{sec:general_potentials} we prove Theorem~\ref{thm:main} using the results from Section~\ref{sec: Discrete harmonic oscillators}. Finally, in Section~\ref{sec: Alternative regions} we discuss the other regimes of the parameter $\gamma$ and provide an overview of the corresponding eigenvalue asymptotics.

Throughout the paper, $C>0$ denotes a generic constant which may change from line to line.

\section{Discrete harmonic oscillators}\label{sec: Discrete harmonic oscillators}

In this section we study the spectrum of the discrete harmonic oscillator. For a vector $\omega\in \R^d$, we let $$V^{harm}:\R^d\to \R,\qquad V^{harm}(x)= \frac{1}2(\omega_1^2x_1^2+\ldots \omega_d^2x_d^2).$$

The continuum harmonic oscillator is given by the operator $H^{harm}=\Delta_{\R^d}+V^{harm}$ on $L^2(\R^d)$ and has purely discrete spectrum with eigenvalues given by 
\begin{align*}
   \sigma(H^{harm})&= \left\{\frac12\sum_{i=1}^d \omega_i(2n_i+1) \mid n_i\in \N_0\right\}
\end{align*}
which we can enumerate as a non-decreasing sequence of eigenvalues counting multiplicities, i.e.,
\begin{align*}
   e_0(V^{harm})\leq e_1(V^{harm})\leq\ldots\leq e_n(V^{harm}){\leq}\ldots.
\end{align*}
%it is well-known that $e_n(V^{harm})$ is the $n$-th eigenvalue of the continuum harmonic oscillator.
We denote the $n$-th eigenvalue of the operator $H_N$ counting multiplicities by $E_n(H_N)$ and show the following eigenvalue asymptotics.
\begin{theorem}\label{thm:eigenvalue_asymptotics_harmonic_oscillator_main}
    For all $n\in \mathbb{N}_0$ and $\gamma\in (-1,\infty)$, it holds
    \begin{align*}
        \lim_{N\to \infty} \frac{E_n(H_N)}{\lambda_N} = e_n(V^{harm}),
    \end{align*}
    where $\lambda_N=N^{1-\gamma}$.
\end{theorem}

Since the harmonic oscillator potential $V^{harm}$ is $2$-homogeneous, we make a further simplification to our setting by moving all of the scaling and constants towards the potential. We furthermore start with the one-dimensional case $d=1$ since the general case then follows by tensorization of the one-dimensional results.

Specifically, for some $\kappa\in (0,1)$, we consider the multiplication operator on $\ell^2(\Z)$ by the function $$v_\kappa:\Z\to \R,\qquad v_\kappa(x)=2\kappa^2V^{harm}(\kappa x)=\kappa^4 x^2.$$ We study the operator 
$$H_\kappa=\Delta + v_\kappa$$
where $\Delta$ is the discrete Laplacian on $\ell^2(\Z)$ defined as
$$(\Delta f)(x)=2f(x)-f(x+1)-f(x-1).$$
The operator $H_\kappa$ has purely discrete spectrum as $\lim_{|x|\to\infty} v_\kappa(x)=\infty$. We denote its $n$-th eigenvalue by $E_n(\kappa)=E_n(H_\kappa)$, $n\in \N_0$.

Now, letting $\kappa=\sqrt{\omega/N^{1+\gamma}}$, one sees that
$$H_N = \frac{N^2}{2}H_\kappa $$ 
for $H_N$ in  $d=1$ and we can derive eigenvalue asymptotics for $E_n(H_N)$ from eigenvalue asymptotics for $E_n(\kappa)$ as $\kappa\to 0$ which are stated in the following theorem.

\begin{theorem}\label{thm:eigenvalue_asymptotics_harmonic_oscillator}
    For all $n\in \mathbb{N}_0$, it holds
    \begin{align*}
        \lim_{\kappa\to 0} \frac{E_n(\kappa)}{\kappa^2} = e_n(V).
    \end{align*}
\end{theorem}

To prove eigenvalue asymptotics for $E_n(\kappa)$ as $\kappa\to 0$, we use the approximate eigenfunctions given by Hermite polynomials. 
Let $h_n$ be the $n$-th Hermite polynomial, $n\in\N_0$, in physicists normalization, i.e., $ h_{n}(y)=(-1)^{n}e^{y^{2}}{\frac {d^{n}}{dy^{n}}}e^{-y^{2}}$, see e.g. \cite[Chapter~22]{Abramowitz_Stegun_72}a nd denote for $y\in \R$,
$$\varphi(y)=e^{-y^2/2},\qquad \Psi_n=h_n\varphi.$$
Furthermore,  for some $\kappa>0$, we define the scaled Hermite polynomial and the corresponding test function  for  $x\in \Z$  
$$\psi_{n}(x)=\Psi_n(\kappa x)=h_{n}(\kappa x)\varphi(\kappa x).$$

%and     $$\kappa_N=\frac{\sqrt{(1-\varepsilon)\lambda_N \omega}}{N}$$
The fundamental property of these test functions is that they are approximate eigenfunctions of the operator $\Delta + v_\kappa$ for small $\kappa$ as shown in the following lemma.

\begin{lemma}\label{lem:approx_eigenfunction}
For all $n \in \mathbb{N}_0$ and $\kappa>0$, it holds
$$(\Delta+v_\kappa)\psi_{n}(x)= \kappa^2 \left(2n+1\right)   \psi_{n}(x)-R(x,\kappa) ,$$
where 
$$R(x,\kappa)=\int_0^\kappa \frac{t^3}{3!} \Psi_{n}^{(4)}(\kappa x-t) dt= O(\kappa^4 )$$
and the $O(\cdot)$ is uniform in $x$. 
\end{lemma}
\begin{proof}
  We expand $\psi_{n}(x\pm 1)= \Psi_{n}(y\pm\kappa)$ about $y=\kappa x$ with respect to $\kappa$
    \begin{multline*}
         \Psi_{n}(y)-\Psi_{n}(y\pm\kappa)\\
         = -\left(\pm \kappa \Psi_{n}'(y)+\frac{\kappa^2}{2}\Psi_{n}''(y)\pm\frac{\kappa^3}{6}\Psi_{n}'''(y)+\int_0^\kappa \frac{t^3}{3!}\Psi_{n}^{(4)}(y-t) dt \right),
    \end{multline*}
    where the integral term is of order $O(\kappa^4)$ and is uniform in $y$ since $\Psi_n$ is a Schwartz function and thus bounded.
    Using that $\Psi_n$ satisfies the harmonic oscillator equation \cite[Section~8.3]{Teschl}
    $$- \Psi_n''(y)+ {y^2} \Psi_n(y)=\left(2n+1\right)\Psi_n(y),$$
    we obtain from the expansion above that
    \begin{align*}
        \Delta \psi_{n}(x)&=  2\psi_{n}(x)-\psi_{n}(x+1)-\psi_{n}(x-1)\\
        &=2\Psi_{n}(y)-\Psi_{n}(y+\kappa)-\Psi_{n}(y-\kappa) \\
        &= - {\kappa^2} \Psi_{n}''(y)+O(\kappa^4 ) 
        = \kappa^2\left(  2n+1  -  {y^2} \right)  \Psi_{n}(y)+O(\kappa^4 )\\
        &= \kappa^2 \left(2n+1 - {\kappa^2 x^2} \right)  \psi_{n}(x)+O(\kappa^4 ).
    \end{align*}
This concludes the proof.
\end{proof}

\subsection{Eigenvalue upper bound for  harmonic oscillators}
In this section we prove an upper bound on the $n$-th eigenvalue $E_n(\kappa)$ of the operator $\Delta + v_\kappa$ via the min-max principle  using the test functions $\psi_n$   defined above.

\begin{proposition}\label{thm:upperbound}
    For all $n\in \mathbb{N}_0$ and $\kappa>0$,   the $n$-th eigenvalue $E_n(\kappa)$ of the operator $\Delta + v_\kappa$ satisfies for  $\kappa$ small 
    $$E_n(\kappa) \leq  \kappa^2 (2n+1) + O(\kappa^3).$$
\end{proposition}

We start by proving the following lemma showing that the test functions $\psi_n$ are almost orthogonal for small $\kappa$. We denote the scalar product on $\ell^2(\Z)$ by $\langle \cdot,\cdot\rangle$.

\begin{lemma}[Upper bound]\label{lem:almost_orthogonality}
 For all $n,m\in \mathbb{N}_0$ and $\kappa>0$, there is $C_{n,m}>0$ such that
\begin{align*}
    \left|\langle \psi_{n },\psi_{m }\rangle- \frac{ \sqrt{\pi}2^nn!}{\kappa }\delta_{n,m} \right| \leq C_{n,m} .
\end{align*}
\end{lemma}
\begin{proof}
 Using the orthogonality of Hermite polynomials with respect to the weight $\varphi^2$, we have
\begin{align*}
    \langle \Psi_{n },\Psi_{m }\rangle_{L^2(\R)} &= \int_{\R} H_n(x)H_m(x)e^{-\omega x^2} dx = \sqrt{\pi}2^nn!\delta_{n,m}.
\end{align*}
Now, we want to estimate the difference to the discrete. To this end, we observe that by Poisson summation formula, we have for any $g$ in the Schwartz space %$\mathcal{S}(\R)$
\begin{align*}
    \sum_{x\in \Z} g(x)=\sum_{k\in \Z} \hat{g}(2\pi k)  = \int_{\R} g(x) dx + \sum_{k\in \Z\setminus\{0\}} \hat{g}(2\pi k). 
\end{align*}
Applying this to $g =\Psi_n \Psi_m $, we have
\begin{align*}
    \langle \psi_{n },\psi_{m }\rangle &=  \sum_{x\in \Z} h_{n}(\kappa  x)h_{m}(\kappa  x) \varphi^2(\kappa x)   \\
    &= \int_{\R} h_n(\kappa x)h_m(\kappa x) e^{-\kappa^2 x^2} dx + \sum_{k\in \Z\setminus\{0\}} \hat{g}(2\pi k) .
\end{align*}
Changing variables $y=\kappa  x$, we have 
\begin{align*}
    \int_{\R} h_n(\kappa x)h_m(\kappa x) e^{-\kappa^2 x^2} dx &= \frac{1}{\kappa} \int_{\R} h_n(y)h_m(y) e^{-y^2} dy = \frac{\sqrt{\pi}2^nn!}{\kappa}\delta_{n,m}.
\end{align*}   
 Finally, we estimate the Fourier coefficients. Note that $\hat{g}$ is the Fourier transform of $\varphi (\kappa x)^2 $ times two Hermite polynomials in $\kappa x$,  and, therefore, $\hat{g}$ is  $\varphi (x/\kappa)^2 /\kappa$ times a linear combination of Hermite polynomials in $x/\kappa$ up to degree $n+m$.
As, $h_l( y) \leq (1+|y|)^l$ for all $l\in \mathbb{N}_0$, $y\in \R$, we have for some constants $C >0$
\begin{align*}
    \sum_{k\in \Z\setminus\{0\}} |\hat{g}(2\pi k)|
\le C\sum_{k\in \Z\setminus\{0\}}  \frac{ k^{n+m}}{\kappa^{n+m+1}}
e^{- ({2\pi k}/{\kappa})^2}
\le  C \int_0^\infty y^{n+m} e^{- y^2} dy\leq C,
\end{align*}
where we estimated the sum by an integral and substituted $y=2\pi k/\kappa$ in the second step. This concludes the proof.
\end{proof}

\begin{proof}[Proof of Proposition~\ref{thm:upperbound}]
   We use the min-max principle to estimate the $n$-th eigenvalue $E_n(\kappa)$ of $H_\kappa=\Delta + \kappa^4 x^2$. To this end, we consider the linear span of the first $n+1$ test functions $\psi_0,\ldots,\psi_n$ which are clearly in the domain of $H_\kappa$ and denote it by $W_n$. For any $\phi\in W_n$, we have $\phi=\sum_{j=0}^n a_j \psi_j$ for some $a_j\in \C$. Using the previous two lemmas, we compute
 \begin{align*}
     \langle \phi, (\Delta + v_\kappa)\phi\rangle &= \sum_{j,m=0}^n a_j \overline{a_m} \langle \psi_j, (\Delta + \kappa^2 x^2)\psi_m\rangle \\
     &= \sum_{j,m=0}^n a_j \overline{a_m} \left(  \kappa^2 (2m+1)  \langle \psi_j,\psi_m\rangle + O(\kappa^4)  \right)\\
     &= \sum_{j=0}^n |a_j|^2  \kappa^2 (2j+1)  \frac{\sqrt{\pi}2^jj!}{\kappa} +  \sum_{j,m=0}^n |a_j|| a_m| O(\kappa^4)\\
     &\leq   \kappa^2 (2n+1) \left( \sum_{j=0}^n |a_j|^2   \frac{\sqrt{\pi}2^jj!}{\kappa} +  \left(\sum_{j=0}^n |a_j|\right)^2 O(\kappa^2)\right).
 \end{align*}
 Furthermore, we have
 \begin{align*}
     \|\phi\|^2 &= \sum_{j,m=0}^n a_j \overline{a_m}  \langle \psi_j,\psi_m\rangle =  \sum_{j=0}^n |a_j|^2  \frac{\sqrt{\pi}2^jj!}{\kappa} +  \left(\sum_{j=0}^n |a_j|\right)^2O(1) .
 \end{align*}   
 Thus, for $\kappa$ small enough, we have
 \begin{align*}
     \frac{\langle \phi, (\Delta + \kappa^2 x^2)\phi\rangle}{\|\phi\|^2} &\leq \kappa^2 (2n+1) \frac{  \sum_{j=0}^n |a_j|^2   \frac{\sqrt{\pi}2^jj!}{\kappa} + O(\kappa^2)  \left(\sum_{j=0}^n |a_j|\right)^2}{  \sum_{j=0}^n |a_j|^2  \frac{\sqrt{\pi}2^jj!}{\kappa} + O(1)  \left(\sum_{j=0}^n |a_j|\right)^2} \\
     &\leq \kappa^2 (2n+1) + O(\kappa^3).
 \end{align*}
 Taking the supremum over all $\phi\in W_n$ %and then the infimum over all $n+1$-dimensional subspaces $V_n$ 
 concludes the proof.    
\end{proof}

\subsection{Eigenvalue upper bound for modified harmonic oscillators}

In this section, we address the lower bound on the eigenvalues which is more involved. Indeed, we have to restrict the operator to suitable intervals and estimate the ground state energies on these intervals. To this end, we modify the harmonic oscillator potential by increasing it outside a certain interval at one point to deal with the approximation error. The result is stated in the following proposition.

\begin{proposition}[Lower bound]\label{thm:lowerbound}
    For all $n\in \mathbb{N}_0$ and $\kappa>0$, the $n$-th eigenvalue $E_n(\kappa)$ of the operator $\Delta + v_\kappa$ satisfies for  $\kappa$ small 
    $$\liminf_{\kappa\to 0} \frac{E_n(\kappa)}{\kappa^2} \geq   2n+1 .$$
\end{proposition}

To prove the lower bound, we use the so called Agmon--Allegretto--Piepenbrink theorem. We first recall that a function on $u:\Z\to\R$ is called an  \emph{$\alpha$-superharmonic} function for an operator $H$ on $C(\Z)$ and  $\alpha \in \mathbb{R}$ if 
\[
(H + \alpha)u \geq 0.
\] 
If $-u$ is $\alpha$-superharmonic, then $u$ is called \emph{$\alpha$-subharmonic} and if both $u$ and $-u$ are $\alpha$-superharmonic, then $u$ is called \emph{$\alpha$-harmonic}.

Assume now that the restriction of $H$ to the compactly supported functions of $\ell^2(\Z)$ is symmetric and  lower semibounded and, thus, allows for a Friedrichs extension which we denote with slight abuse of notation also by $H$. Let $E_0$ be  the ground state energy of $H$, i.e., the bottom of the spectrum. The Agmon--Allegretto--Piepenbrink criterion, see e.g. \cite[Theorem~4.14]{Keller_Lenz_Wojciechowski_21},  connects the existence of positive superharmonic functions to lower bounds on the ground state energy for some $\alpha\in \R$ in the following equivalence: 
\begin{itemize}
    \item [(i)] It holds $\alpha \geq -E_0$.
    \item [(ii)] There exists a strictly positive $\alpha$-superharmonic function.
\end{itemize}

Clearly, the Agmon--Allegretto--Piepenbrink theorem is only suitable to estimate the ground state energy. However, we can use it to estimate higher eigenvalues by decomposing and restricting our operator to suitable intervals. 
These intervals are chosen with respect to zeros of the Hermite polynomials. We then estimate the ground state energy of the restricted operator on these intervals using the   Agmon--Allegretto--Piepenbrink theorem and use the ground states on these intervals to construct an approximate eigenfunction functions for the operator on $\Z$. 

To estimate the ground state energies on these intervals, we again use the approximate eigenfunctions given by the Hermite polynomials. In the interior of these intervals, these functions are indeed approximate eigenfunctions by Lemma~\ref{lem:approx_eigenfunction}
$$(\Delta+v_\kappa)\psi_{n}(x)= \kappa^2 \left(2n+1\right)   \psi_{n}(x)+O(\kappa^4 ) .$$ 
However, at the boundary of these intervals, the Hermite polynomials vanish across an edge which causes problems. These problems manifest themselves, on the one hand, for the bounded intervals between two zeros and also for the unbounded interval at infinity since the error term $O(\kappa^4)$ Lemma~\ref{lem:approx_eigenfunction} is independent of $\psi_n(x)$ which is a problem if $\psi_n(x)$ is small.

Here, we employ two different strategies to control the behavior of $\psi_n(x)$ near the boundaries of the intervals:
\begin{itemize}
    \item First, we prove a lower bound on $\psi_n(x)$ for $x$ {such that $\kappa x$ has}  distance at least $\kappa$ from the zeros of the Hermite polynomial, see Lemma~\ref{lem:lower_bound_hermite} below.
    \item Second, we modify the functions $\psi_n$ near the interval boundaries by shifting and rescaling their arguments so that a zero of the Hermite polynomial is located exactly at the boundary vertex outside the interval. At the opposite boundary, the function $\psi_n$ retains a fixed sign and does not vanish, which allows us to apply the approximate eigenfunction property from Lemma~\ref{lem:approx_eigenfunction} directly.
\end{itemize}
For the unbounded interval, the situation is different. A closer inspection of the error term $O(\kappa^4)$ in Lemma~\ref{lem:approx_eigenfunction} shows that, although it also contains a Gaussian factor, the associated polynomial factor is larger than that of $\psi_n$ and appears with a negative sign. Consequently, for large $x$, the approximate eigenfunction property cannot be applied directly. We address this issue by increasing the potential of the harmonic oscillator. We then show that, for fixed $n$, the appropriately scaled eigenvalues of the harmonic oscillator and of the modified oscillator converge to the same value, since the corresponding eigenfunctions localize near zero.

To pursue these ideas we now discuss how to restrict the operator to intervals and then construct the intervals such that the approximate eigenfunctions vanish on one side of the boundary.

For an interval $I\subseteq \Z$,  we consider is the projection $$\pi_{I}\colon \ell^2(\mathbb{Z})\rightarrow \ell^2(I),\quad f\mapsto f\vert_I$$ onto $\ell^2(I)$ and $$\iota_{I}\colon \ell^2(I)\rightarrow \ell^2(\mathbb{Z}),\quad f\mapsto
\begin{cases}
f, & \mbox{on } I,\\
0, & \mbox{else}.
\end{cases}$$ 
its adjoint  which is the canonical injection by continuation by zero.

Now, we construct the specific intervals for fixed  $n\in \N_0$. Observe that the Hermite polynomial $h_n$ has exactly $n$ distinct real zeros which are symmetric about $0$ and is even or odd depending on whether $n$ is even or odd.
Let $0\leq z_1<z_2<\ldots < z_k$ be the non-negative zeros of the $n$-th Hermite polynomial, where 
$$k=\begin{cases}
n/2  & \text{ if $n$ is even, i.e., $z_1>0$},\\
(n+1)/2 & \text{ if $n$ is odd, i.e., $z_1=0$}.
\end{cases}
$$ 
Moreover, we set $z_{k+1}=\infty$ for convenience. By symmetry $-z_k>\ldots -z_1\ge 0$ are zeros of $h_n$ as well.

The starting point of the interval $I_1=I_{1,n}(\kappa)=[a_1,b_1]\cap \Z$ is defined as
$$ a_1 =a_{1,n}(\kappa)=\left\lfloor{\frac{z_1}{\kappa}}\right\rfloor+1.$$
If  $z_1=0$ (in the case where $n$ is odd), then  $a_1=1$  and we set $\beta_1=1$. Otherwise, we set
$$\beta_1=\beta_{1,n}(\kappa)=\frac{z_1}{\kappa (a_1-1)}. $$
%Then, we have also $\psi_n(\beta_1(a_1-1))= h_n(z_1)\varphi(z_1)=0$.
We then define $$b_1=b_{1,n}(\kappa) = \left\lfloor\frac{z_{2}}{\beta_1\kappa}\right\rfloor-1.$$

We proceed inductively. Given $a_1\leq b_1\leq \ldots\leq a_{j-1}\leq b_{j-1}$, and
$\beta_1,\ldots,\beta_{j-1}  $, we let $$I_j=I_{j,n}(\kappa)=[a_j,b_j]\cap \Z$$ be given by
\begin{align*}
a_j=a_{j,n}(\kappa) &=\left\lfloor\frac{z_j}{\beta_{j-1}\kappa}\right\rfloor+1,\\
\beta_j=\beta_{j,n}(\kappa) &=\frac{z_j}{\kappa (a_j-1)},\\  
 b_j=b_{j,n}(\kappa) &=\left\lfloor\frac{z_{j+1}}{\beta_j\kappa}\right\rfloor-1,
\end{align*} 
where we set $b_{k}=\infty$ and, thus, $$I_k=I_{k,n}(\kappa)=[a_k,\infty)\cap \Z.$$
If $n$ is even, we further define $b_0=b_{0,n}(\kappa)=a_1-2$
\begin{align*}
    I_0= I_{n,0}(\kappa)=[-b_0,b_0]\cap \Z=
    \left[{-\left\lfloor{\frac{z_1}{\kappa}}\right\rfloor+1 },\left\lfloor{\frac{z_{1}}{\kappa}}\right\rfloor-1  \right] \cap \Z, 
\end{align*}
and $I_0=\emptyset$ if $n$ is odd (i.e., $z_1=0$). Finally, for $j=1,\ldots,k$, we define the negative intervals via
\begin{align*}
    I_{-j}= I_{n,-j}(\kappa)=[-b_j,-a_j]\cap \Z
\end{align*}
and set $\beta_{-j}=\beta_j$.

We summarize the properties of these intervals in the next lemma.
\begin{lemma}[Properties of the intervals]\label{lem:intervals}
    Let $n\in \N_0$ and $\kappa>0$. Then, the intervals $I_j=I_{j,n}(\kappa)$, $j=-k,\ldots,k$ defined above are disjoint and satisfy
   \begin{align*}
    \Z = \bigcup_{j=-k}^{k} I_j \,\dot\cup   \, \left\{\pm\left\lfloor {\frac{z_1}{\beta_1\kappa}  }\right\rfloor , \ldots,\pm\left\lfloor  {\frac{z_k}{\beta_k\kappa}  }\right\rfloor \right\}.
\end{align*} 
Furthermore, for all $j=1,\ldots,k$, it holds
\begin{align*}
    \psi_n(\beta_j (a_j-1))=0=\psi_n(-\beta_{-j} (a_j-1))
\end{align*}
and
$$1\le \beta_j=1+O(\kappa).$$ 
\end{lemma}
\begin{proof}
The disjointness of the intervals follows directly from their construction. Furthermore, by construction we have for all $j=1,\ldots,k$ that $b_j$ and $a_{j+1}$ have distance exactly $2$ with the point $\left\lfloor{z_{j+1}}/{\beta_j\kappa}\right\rfloor$ in between. Thus, the intervals $I_j$, $j=-k,\ldots,k$ cover all of $\Z$ except for the points $\pm\left\lfloor {{z_j}/{\beta_j\kappa}  }\right\rfloor$ for $j=1,\ldots,k$. This gives the desired decomposition of $\Z$. 
The fact that $ \psi_n$ vanishes at the points $\pm\beta_j (a_j-1)$ follows
since
$$\psi_n(\beta_j (a_j-1))=h_n(z_j)\varphi(z_j)=0.$$
As for the final statement   $\beta_j\ge 1$ is clear by definition and the upper bound  follows by induction. {Let us introduce the fractional part of a real number by $x = \floor{x} + \{x\}$, where $0\leq \{x\}\leq 1$. For $j=1$, we have by construction since $a_1=\left\lfloor{z_1}/{\kappa}\right\rfloor+1$ that
$$1\leq \beta_1=\frac{z_1}{\kappa (a_1-1)}  = \frac{z_1}{z_1 - \kappa\left\{ z_1 \kappa^{-1} \right\}}\le  {\frac{z_1}{z_1 - \kappa}} =1+O(\kappa).$$
If the statement holds for $j-1$, then we have since $a_j=\left\lfloor{z_j}/{\beta_{j-1}\kappa}\right\rfloor+1$  
$$\beta_j=\frac{z_j}{\kappa (a_j-1)}=\beta_{j-1}\frac{z_j}{z_j - \kappa \beta_{j-1}\{z_j (\kappa\beta_j)^{-1}\}}=1+O(\kappa).$$} 
This concludes the proof.
\end{proof}

We define the restricted operators
\begin{align*}
    H_{j,\kappa}  = \pi_{I_j} (\Delta + v_\kappa)\iota_{I_j}
\end{align*}
on $\ell^2(I_j)$  which extend to operators on $C(I_j)$ for $j=-k,\ldots,k$. We will show below  that the ground state energies of these operators satisfy the desired lower bound. To do so, the lemma above suggests defining the test functions $\psi_{n,j}$ on $I_j$ via
    $$\psi_{n,j}(x)= |\psi_n(\beta_j x)|= |\Psi(\beta_j\kappa x)|= |h_n(\beta_j\kappa x)| \varphi(\beta_j\kappa x).$$

We first prove the bounds for  $j=-(k-1),\ldots,k-1$. As described above, we first we need a lower bound on  $\psi_n$ away from the zeros of the Hermite polynomial. This is proven in the next lemma.

%We explain the strategy for $j>0$, the case $j<0$ follows then by summetry and the case $j=0$ for odd $n$ is similar. Note that on the interval $I_j$, the Hermite polynomial $H_n$ does not change sign and thus $\psi_n$ is either positive or negative there. We will construct an $\alpha$-superharmonic function for $\alpha$ slightly below $\kappa^2(2n+1)$ by modifying the restriction of $\psi_n$ to $I_j$ so that it becomes zero at the left boundary of the interval. On the right boundary, there is no issue as $\psi_n$ has the same sign as it is on the next point outside the intervall (which is still less than the zero of the Hermite polynomial). We then want to employ Lemma~\ref{lem:approx_eigenfunction},to that this modified function is indeed $\alpha$-superharmonic for $\alpha$ slightly below $\kappa^2(2n+1)$ for $\kappa$ small enough. However, to this end we need a lower bound on  $\psi_n$ away from the zeros of the Hermite polynomial. This is proven in the next lemma.

\begin{lemma}[Behavior near zeros]\label{lem:lower_bound_hermite}
    For all $n\in \N$, there is $c_n>0$ such that for all $\kappa>0$ small enough and all $j=-(k-1),\ldots,k-1$, it holds
    $$\min_{x\in I_j} |\psi_{n,j}(x)| \geq c_n \kappa\quad\mbox{and}\quad |\psi_{n,\pm k}(a_{\pm k})| \geq c_n\kappa. $$
\end{lemma}
\begin{proof}
By symmetry we only have to check the statement for positive $j=1,\ldots,k$.
One directly checks that for $j=1,\ldots,k$, $i=0,\ldots,k-1$
\begin{align*}
    | \beta_j\kappa a_j- z_j  |= z_j\frac{1}{a_j-1}\ge \beta_{j-1}\kappa \ge\kappa,
\end{align*}\begin{align*}
     | z_{i+1}-\beta_i\kappa b_i   |%&=   z_{i+1}-\beta_i\kappa (a_{i+1}-2)  \ge   
     =\beta_i\kappa+z_{i+1} -\beta_{i+1}\kappa(a_{i+1}-1)   = \beta_{i}\kappa \ge\kappa.
\end{align*}
So, in particular, for $x\in I_j$, $j=-(k-1),\ldots,k-1$, we have that $| \beta_j\kappa x- z  | \geq {\kappa}$ for $z=z_j$ or $z=z_{j+1}  $.

To estimate $\psi_{n,j}$, we recall that $\psi_{n,j}(x)=h_n(\beta_j\kappa x)\varphi(\beta_j\kappa x)$ and that the Gaußian $\varphi$ has a uniform lower bound on the bounded intervals $I_j$, $j=1,\ldots,k-1$, say $c>0$. 
Moreover, by continuity of $h_n$, it suffices to check the lower bound only at the vertices near the  zeros of $h_n$ which are $x=a_j$ and $x=b_j$.
We expand $h_n$ about $z=z_j$ or $z=z_{j+1}$ to obtain for $x=a_j$ and $x=b_j$ respectively, using $h_n(z)=0$,
\begin{align*}
    h_n(\beta_j\kappa x)%&= h_n(z) + h_n'(z)(\beta_j\kappa x - z) + O((\beta_j\kappa x - z)^2) \\
    &= h_n'(z)(\beta_j\kappa x - z) + O((\beta_j\kappa)^2).
\end{align*}
It is well-known that the zeros of $h_{n-1}$ are distinct from the ones of $h_n$ \cite[Section 22.16(iii)]{Abramowitz_Stegun_72}, \cite[Theorem 3.3.1]{Szego}.  Thus, by the recursion formula for Hermite polynomials, we have
$$h_n'(z) = 2n h_{n-1}(z) \neq 0.$$
Above we have checked that $|\beta_j\kappa x - z|\ge \kappa$. So,  putting $c_n= |h_{n-1}(z)| >0$, we obtain the result for $\kappa$ small enough.
\end{proof}

\begin{lemma}[Bounded intervals]\label{lem:lower_bound_interval}
    For   $n\in \N$ and $j=-(k-1),\ldots,k-1$,  the eigenvalues $E_{n,j}(\kappa)$ of the operator $H_{j,\kappa}$ satisfies for  $\kappa$ small 
    $$\liminf_{\kappa\to 0}\frac{E_{n,j}(\kappa)}{\kappa^2} \geq   (2n+1).$$
\end{lemma}
\begin{proof}
    The strategy is to show that $\psi_{n,j}$ is an $\alpha$-superharmonic function for $\alpha = (1-\varepsilon)\kappa^2(2n+1) $ for given $\varepsilon>0$ and $\kappa$ chosen small enough accordingly.
    
   By Lemma~\ref{lem:approx_eigenfunction}, $\psi_{n,j}$ satisfies for $x\in I_j\setminus \{b_j\}  $,
   \begin{align*}
      H_{j,\kappa}\psi_{n,j}(x)&=    (\Delta + v_{\beta_j\kappa})\psi_{n,j}(x) + (v_{\kappa}- v_{\beta_j\kappa})\psi_{n,j}(x) \\
         &
         = \beta_j^2 \kappa^2 (2n+1)  \psi_{n,j}(x)+ {\kappa}^4(1-\beta_j^4)x^2\psi_{n,j}(x)+ O(\kappa^4).
   \end{align*} 
   In particular,   the equation is satisfied at $x=a_j$ as $\psi_{n,j}(a_j-1)=0$ because of the choice of $\beta_j$. Now,  $\psi_{n,j}(x)>c_n$ by Lemma~\ref{lem:lower_bound_hermite} above, as well as $\beta_j=1+O(\kappa)$ by Lemma~\ref{lem:intervals}, and $x^2\le b_j^2$. Thus, for given $\varepsilon>0$, we can choose $ \kappa>0$ small enough such that
   $$ ( H_{j,\kappa} -(1-\varepsilon)\kappa^2(2n+1) )\psi_{n,j}(x) \ge 0.$$
   We are left to check the case of $x=b_j$. We observe that $\beta_j\kappa(b_j+1) \le z_{j+1}$. Thus,  $h_n(\beta_j\kappa b_j)$ and  $h_n(\beta_j\kappa(b_j+1))$ have the same sign and so do $\psi_n(\beta_jb_j)$ and  $\psi_n(\beta_j(b_j+1))$ which we assume without loss of generality to be positive.
   Hence, as $\psi_{n,j}$ vanishes outside of $I_j$
   $$  \psi_{n,j}( b_j)-\psi_{n,j}( b_j+1)=\psi_{n,j}(b_j)=\psi_n(\beta_jb_j) \ge 
   \psi_n(\beta_jb_j)-\psi_n(\beta_j(b_j+1)). $$
   We obtain from Lemma~\ref{lem:approx_eigenfunction}
\begin{align*}
      H_{j,\kappa}\psi_{n,j}(b_j)&\ge    (\Delta + v_{\beta_j\kappa})\psi_{n,j}(b_j) + (v_{\kappa}- v_{\beta_j\kappa})\psi_{n,j}(b_j) \\
         &
         = \beta_j^2 \kappa^2 (2n+1)  \psi_{n,j}(b_j)+ {\kappa}^4(1-\beta_j^4)b_j^2\psi_{n,j}(b_j)+ O(\kappa^4)
   \end{align*} 
   from which we conclude $ ( H_{j,\kappa} -(1-\varepsilon)\kappa^2(2n+1) )\psi_{n,j}(x) \ge 0$ for small $\kappa$ as above. Since $\varepsilon >0$ can be chosen arbitrarily small, we conclude the result.
\end{proof}

%\begin{remark}    \label{r:unboundedinterval}
    %The above strategy does not work for the unbounded interval $I_k$ directly. The reason is that the error term $O(\kappa^4)$ in Lemma~\ref{lem:approx_eigenfunction} is independent of $\psi_n(x)$ and, thus, for $x$ large, the term ${\kappa}^4(1-\beta_j^4)x^2\psi_{n,j}(x)+ O(\kappa^4)$ is negative and larger in absolute value than $\beta_j^2 \kappa^2 (2n+1)  \psi_{n,j}(x)$ which prevents us from concluding that $\psi_{n,j}$ is an $\alpha$-superharmonic function for some $\alpha$. More precisely, inspecting the proof of Lemma~\ref{lem:approx_eigenfunction} and  we see that the error term is given by
%\begin{multline*}    -\int_0^\kappa \frac{t^4}{4!}\Psi_{n}^{(4)}(y-t) dt =-\int_0^\kappa \frac{t^4}{4!}\sum_{j=1}^4\binom{4}{j}(h_n^{(4-j)}\varphi^{(j)})  (y-t) dt \\    =-\int_0^\kappa {t^4}\varphi   (y-t)\sum_{j=0}^4c_{n,j} (h_{n-4+j}h_j)  (y-t) dt \sim \, 
%&= \frac{\kappa^4}{4!} \Psi_{n}^{(4)}(y) + O(\kappa^5) \\\end{multline*}
%where we used the recursion formul $h_k'=2kh_{k-1}$ for Hermite polynomials in the last step. The dominant term in the sum is for $j=4$ which gives rise to the term $$\kappa ^4(\varphi h_n h_4) (x)=\kappa ^4(h_4\psi_n)(x)\sim \kappa^8 x^4 \psi_n (x)$$ in the error above which is also negative for $x$ large enough. This shows that the error term grows faster in $x$ than the leading term $\kappa^2(2n+1)\psi_n(x)$ which prevents us from concluding the desired superharmonicity.\end{remark}

For the unbounded interval $I_k$, we have an additional issue because $\psi_{n,k}$ becomes superexponentially small and, indeed, as will be clear from the proof below small than the error $O(\kappa^4)$. So, we need to modify the potential and later also the test function. We will later show that this modification does not change the asymptotic behavior of the eigenvalues for fixed $n$ as $\kappa\to 0$.
To this end, we define for $\delta>1$ the modified potential
$$
\tilde v_{\kappa}(x)=\begin{cases} v_{\kappa}(x)=\kappa^4 x^2, & |x|\neq \floor{\kappa^{-(1+\delta)}},\\
 \kappa^{-\delta}, & |x|=\floor{\kappa^{-(1+\delta)}} .
\end{cases}
 $$
We then define the modified restricted operators
	\begin{align*}
		\tilde H_{  \pm j,\kappa}   = \pi_{I_{ \pm j}} (\Delta + \tilde v_{\kappa})\iota_{I_{ \pm j}}
	\end{align*}
	and their corresponding eigenvalues by $\tilde E_{0,j}(\kappa)$, for $j = -k,\dots , k$. We observe that $ H_{\pm j,\kappa}=\tilde H_{\pm k,\kappa}$ for $|j|< k$ and, thus, $ E_{n,\pm j}(\kappa)=\tilde E_{n,\pm j}(\kappa)$. 

\begin{lemma}[Unbounded intervals]\label{lem:lower_bound_modified} Let $\delta\in(0,1/2)$.
    For all  $n\in \N$,  the eigenvalues $E_{n,  k}(\kappa)$ of the operator $\tilde H_{k,\kappa}$ satisfies for  $\kappa$ small 
    $$\lim_{\kappa\to 0}\frac{\tilde E_{n,k}(\kappa)}{\kappa^2} \geq   2n+1.$$
\end{lemma}
\begin{proof}
We define the test functions %for $x_\delta=\floor{\kappa^{-(1+\delta)}}$,
$$\tilde \psi_{n,k}=
\begin{cases} \psi_{n,k}(x), & x\leq x_\delta=\floor{\kappa^{-(1+\delta)}},\\
\psi_{n,k}(x_\delta), & x> x_\delta=\floor{\kappa^{-(1+\delta)}} .\end{cases}$$
For $x\in I_k$ with $x<\kappa^{-(1+\delta)}$, we take a closer look at the error term $R(x,\kappa)$ in Lemma~\ref{lem:approx_eigenfunction} 
\begin{multline*}
  R(x,\kappa)=  -\int_0^\kappa \frac{t^3}{3!}\Psi_{n}^{(4)}(\kappa x-t) dt\\ =-\int_0^\kappa \frac{t^3}{3!}\sum_{j=1}^4\binom{4}{j}(h_n^{(4-j)}\varphi^{(j)})  (\kappa x-t) dt \\
    =-\sum_{j=0}^4c_{n,j}\int_0^\kappa t^3\varphi   (\kappa x-t) (h_{n-4+j}\tilde h_j)  (\kappa x-t) dt  \, 
%&= \frac{\kappa^4}{4!} \Psi_{n}^{(4)}(y) + O(\kappa^5) \\
\end{multline*}
where, in the last step, we used the recursion formula $h_k'=2kh_{k-1}$ for Hermite polynomials as well as $\varphi^{(j)}=(-1)^j \tilde h_j \varphi$ (where $\tilde h_j$ are the Hermite polynomials of degree $j$ in probabilistic normalization) 
and $c_{j,n}$ are constants. Since the Hermite polynomials $h_l,\tilde h_l$ are of degree $l$, the dominant term in the sum is the one with $j=4$ which gives rise to the estimate 
$$
 R(x,\kappa)\ge -C \kappa ^4(\varphi h_n \tilde h_4) (x)=-C\kappa ^4(\tilde h_4\psi_n)(x)\ge -C \kappa^8 x^4 \psi_n (x).$$   
We now recall that $\tilde\psi_{n,k}(x)=\psi_{n,k}(x)$ for $x\le x_\delta=\floor{\kappa^{-(1+\delta)}}$ and $\tilde v_\kappa(x)=v_\kappa(x)$ for $x< x_\delta$.  Thus,
we have for $x\in I_k$ with $x<x_\delta$
 as above  with the refined estimate on the error in Lemma~\ref{lem:approx_eigenfunction} and using $(v_\kappa-v_{\beta_k\kappa})(x)=\kappa^2(1-\beta_k^2)x^2$  and  
 $\beta_k=1+O(\kappa)$ from Lemma~\ref{lem:intervals} that
 \begin{multline*}
 ( \tilde H_{k,\kappa} -(1-\varepsilon)\kappa^2(2n+1) )\tilde\psi_{n,j}(x) =(  H_{k,\kappa} -(1-\varepsilon)\kappa^2(2n+1) ) \psi_{n,j}(x)\\\ge \varepsilon\beta_j^2 \kappa^2 (2n+1)  \psi_{n,k}(x)+ {\kappa}^4(1-\beta_k^4)x^2 \psi_{n,k}(x) - C \kappa^8 x^4 \psi_{n,j}(x)\\
 \ge  {\kappa}^2  \left(\varepsilon  (2n+1)  -C( {\kappa}^3x^2 + \kappa^6 x^4)\right)  \psi_{n,j}(x)\\
{\kappa}^2  \left(\varepsilon  (2n+1)  -C ({\kappa}^{1-2\delta} + C \kappa^{2-4\delta})\right)  \psi_{n,j}(x) \ge 0
 \end{multline*}  
 for $\kappa$ small enough since $\delta<1/2$. Observe that the  estimate, in particular, holds at $x=a_k$ since $\tilde\psi_{n,k}(a_k-1)=0$ because of the choice of $\beta_k$ and the arguments presented in Lemma~\ref{lem:lower_bound_interval} near the zeros of $h_n$. 
 For $x_\delta=\floor{\kappa^{-(1+\delta)}}$, we have by construction of $\tilde \psi_{n,k}$ that
 \begin{multline*}
    \tilde H_{k,\kappa} \tilde\psi_{n,j}(x_\delta) =  H_{k,\kappa} \psi_{n,k}(x_\delta)+(\tilde v_\kappa\tilde\psi_{n,k}-v_\kappa\psi_{n,k})(x_\delta)+ (\psi_{n,k}-\tilde\psi_{n,k})(x_\delta+1).
 \end{multline*}
 By the same argument as above, we have
 $$ (H_{k,\kappa} -(1-\varepsilon)\kappa^2(2n+1) ) \psi_{n,k}(x_\delta)\ge0$$
 for $\kappa$ small enough. So, it remains to estimate the last two terms.
 Since $\tilde\psi_{n,k}(x_\delta)=\psi_{n,k}(x_\delta)$, we have
 \begin{align*}
    (\tilde v_\kappa\tilde\psi_{n,k}-v_\kappa\psi_{n,k})(x_\delta)
    =(\kappa^{-\delta}-\kappa^{2(1-\delta)}) \psi_{n,k}(x_\delta)\ge \kappa^{-\delta/2} \psi_{n,k}(x_\delta)
 \end{align*}
 for $\kappa$ small enough. Moreover, since $h_n(\beta_k\kappa x)>0$ for $x\in I_k$, we have $\psi_{n,k}>0$ on $I_k$ and, thus,
\begin{align*}
   (\psi_{n,k}-\tilde\psi_{n,k})(x_\delta+1) =\psi_{n,k}(x_\delta+1)-\psi_{n,k}(x_\delta)\ge -\psi_{n,k}(x_\delta).
 \end{align*}
Combining these estimates with the refinements of the approximate eigenfunction property above, we obtain for $x_\delta=\floor{\kappa^{-(1+\delta)}}$
 \begin{align*}  
 ( \tilde H_{k,\kappa} -(1-\varepsilon)\kappa^2(2n+1) ) \tilde\psi_{n,j}(x_\delta)\ge \left(\kappa^{-\delta/2} -1\right)\psi_{n,k}(x_\delta) \ge 0
 \end{align*}  
 for $\kappa$ small enough. 
 
 Finally, for $x>x_\delta=\floor{\kappa^{-(1+\delta)}}$, we have $\tilde\psi_{n,k}(x)=\psi_{n,k}(x_\delta)$ and $\tilde v_\kappa(x)=v_\kappa(x)=\kappa^4 x^2$. Thus, 
 \begin{align*}  
    \tilde H_{k,\kappa} \tilde\psi_{n,j}(x) &= \tilde v_\kappa(x)\tilde\psi_{n,k}(x)
    \ge \kappa^2 \kappa^{-2-\delta} \tilde\psi_{n,k}(x_\delta)\ge \kappa^2(2n+1) \tilde\psi_{n,k}(x).
 \end{align*}    
 This concludes the proof as $\varepsilon>0$ can be chosen arbitrarily small.
\end{proof}

Combining the lemmas above,  we will now prove the eigenvalue lower bound for the modified oscillator. To do so, we first prove that an eigenfunction to the $n$-th eigenvalue must have  $n+1$ graph nodal domains, i.e., {maximal connected} regions where it does not change sign.

We start by showing that for symmetric potentials eigenvalues can have  at multiplicity two.

\begin{lemma}\label{lem:degenracy} Let $H=\Delta+v$ be a discrete Schr\"odinger operator on  $\ell^2( \Z)$ with purely discrete spectrum. Suppose that \begin{equation*}
  	v(x) = v(-x).
  \end{equation*}
 Then, all eigenfunctions are either symmetric or anti-symmetric and,  eigenvalues can have at most multiplicity two. Moreover, if an eigenvalue has multiplicity two, then one of the corresponding eigenfunctions is symmetric and the other one is anti-symmetric. 
\end{lemma}
\begin{proof}
    Let $ f $ be an eigenfunction of $H$. 
    We first observe that if $f$ vanishes at zero, then $f$ is anti-symmetric. Indeed, by the eigenvalue equation at zero, we have $f(1)=f(-1)$ and, thus, by induction $f(x)=f(-x)$ for all $x\in \Z$. 
    
    Otherwise, we can define the symmetric and anti-symmetric parts of $f$ by
$$f_s(x)=\frac{f(x)+f(-x)}{2},\qquad f_a(x)=\frac{f(x)-f(-x)}{2}$$
which are clearly eigenfunctions of $H$ with the same eigenvalue as $f$. If $f$ is not symmetric or anti-symmetric, then $f_s$ and $f_a$ are linearly independent and, thus, the eigenvalue corresponding to $f$ has multiplicity at least two.

Assume now that there is are two symmetric eigenfunctions $f$ and $g$ corresponding to the same eigenvalue. As observed above, $f(0)\neq 0$ and $g(0)\neq 0$ since they are symmetric. Hence, we can assume without loss of generality that $f(0)=g(0)$ by rescaling one of the two eigenfunctions. Now, the eigenvalue equation at zero implies that $f(-1)=f(1)=g(1)=g(-1)$. By induction, we obtain that $f(x)=g(x)$ for all $x\in \Z$.

The same argument applies if there are two anti-symmetric eigenfunctions corresponding to the same eigenvalue since the eigenfunction is inductively determined by its values $f(0)=0$ at $f(1)=-f(-1)$. This concludes the proof.
\end{proof}

We  make the following convention for the ordering of the eigenvalues and eigenfunctions of a symmetric discrete Schr\"odinger operator.

\begin{convention}
For a Schrödinger operator with discrete spectrum and symmetric potential we arrange 
 the eigenvalues $\lambda_n$ counted with multiplicity  in increasing order and choose the corresponding eigenfunctions $f_n$ to be either symmetric or anti-symmetric.  
If an eigenvalue has multiplicity two, i.e., $\lambda_n = \lambda_{n+1}$, then we choose the symmetric eigenfunction to be $f_n$ and the anti.symmetric to be  $f_{n+1}$.
\end{convention}

We want to apply a result for half-line operators  \cite[Corollary 4.9]{Teschl2000} to the symmetric operator on the full line. This result states that if  a generic  Jacobi operator \begin{equation*}
	H f(x) = a(x)f(x)-b(x-1)f(x-1) -b(x+1) f(x+1)
\end{equation*}
defined on  $\ell^2(\mathbb{N}_0)$ has discrete spectrum with each eigenvalue having multiplicity at most one, then the $n$-th eigenfunction 
 has exaclty $n+1$ graph nodal domains.

\begin{lemma}\label{lem: zeros of excited states}
  Let $H=\Delta+v$ be a discrete Schr\"odinger operator on  $\ell^2( \Z)$ with purely discrete spectrum and symmetric potential 
  \begin{equation*}
  	v(x) = v(-x).
  \end{equation*} If $f_n$ is the $(n+1)$-th eigenfunction of $H$, the $f_n$ has exactly $n+1$ graph nodal domains. In particular, if $f$ is an eigenfunction with $n+1$ zeros, then it is the $(n+1)$-th eigenfunction of $H$.
  \end{lemma}
\begin{proof}
    We first assume that $f_s$ is an  symmetric eigenfunction to the eigenvalue $\lambda$. The eigenvalue equation for $f_s$ gives that the restriction of $f_s$ to $\mathbb{N}_0$ satisfies $\hat H_+ f_s = \lambda f_s$ where $\hat H_+$ is the non-symmetric Jacobi operator defined on $\ell^2(\mathbb{N}_0)$ by
\begin{align*} 
    & \hat H_+ f(0) = 2f(0)-2f(1) + v(0)f(0) \\
&\hat H_+ f(x) = 2f(x)-f(x-1)-f(x+1) + v(x)f(x),\qquad x\ge 1.
    \end{align*}
  However, to apply the result from \cite{Teschl2000}, we need a self-adjoint Jacobi operator. We can obtain such an operator by rescaling the above operator as follows. We define $ H_+$ on $\ell^2(\mathbb{N}_0)$ by  
  \begin{align*}
		&{H_+} f(0) = 2f(0)-\sqrt{2}f(1)+v(0)f(0),\\
		&{H_+} f(1) = 2f(1) -\sqrt{2}f(0)-f(2)+v(1)f(1),\\
		&{H_+} f(x) = 2f(x)- f(x-1)-f(x+1)+v(x)f(x),\qquad x\ge 2.
	\end{align*} 
    By direct computation,
    it can be seen that if $f$ on $\mathbb{N}_0$ satisfies $\hat H_+ f = \lambda f$, then $g$ defined by $$g(0) = f(0)/\sqrt{2},\qquad g(x) = f(x),\quad x\ge 1$$ satisfies $$ H_+ g = \lambda g. $$ 
    Moreover, the rescaling factor is strictly positive, hence the sign pattern of $f$ and $g$ coincides, thus their graph nodal domains agree. 

    Thus, we see that for any symmetric eigenfunction $f_s$ of $H$, there is a corresponding eigenfunction $g_s$ of the symmetric Jacobi operator $H_+$ with the same eigenvalue and the same graph nodal domains. On the other hand, if $g_s$ is an eigenfunction of $H_+$, then the function $f_s$ defined by $f_s(0) = \sqrt{2}g_s(0)$ and $f_s(x) = g_s(x)$ for $x>0$ is a symmetric eigenfunction of $\hat H$ with the same eigenvalue and the same graph nodal domains.

    By the lemma above there is at most one symmetric eigenfunction corresponding to a given eigenvalue, hence the eigenvalues of $H_+$ have multiplicity at most one. We order the symmetric eigenfunctions $f_{s,n}$ of $H$ and the corresponding eigenfunctions $g_{s,n}$ of $H_+$ in increasing order of their eigenvalues.

Applying  \cite[Corollary 4.9]{Teschl2000} to the symmetric Jacobi operator $H_+$ on $\mathbb{N}_0$, we obtain that the $n+1$-th eigenfunction $g_{s,n}$ of $H_+$ has exactly $n+1$ graph nodal domains.
 Thus, on account of the reflection symmetry, the corresponding  eigenfuntion $f_{s,n}$ on $\ell^2(\Z)$ has $2n+1$ nodal domains.

The argument for anti-symmetric eigenfunctions is analogous and varies only in the following details: Since the eigenvalue equation at zero is trivial, we obtain the correpsondence to a half-line operator on $ \N $ rather than $ \N_0 $ which looks however the same after shifting the index by one. Consequently, the rescaling factor is $\sqrt{2}$ at $x=1$ rather than at $x=0$. So, there is a one-to-one correspondence between anti-symmetric eigenfunctions $f_{a,n}$ of $H$ and eigenfunctions $g_{a,n}$ of the symmetric Jacobi operator on $\ell^2(\mathbb{N})$. Application of \cite[Corollary 4.9]{Teschl2000} to the latter operator gives that $g_{a,n}$ has $n+1$ graph nodal domains and, thus, $f_{a,n}$ has $2n+2$ graph nodal domains.

All that is left to prove to conclude is that if $n$ is even (odd) the $(n+1)$-th eigenfunction of $H$ is symmetric (anti-symmetric). We prove this by induction. The case $n=0$ is trivial since the ground state is necessarily symmetric. Suppose that the statement is true for $n$. First suppose that $n$ is odd and $n+1$ is even -- the other case is analogous. Suppose by contradiction that $f_{n+1}$ is anti-symmetric. Then, we know by induction that there are $(n+1)/2$ anti-symmetric eigenfunctions with eigenvalues smaller than $\lambda_{n+1}$. Thus, by the first part of this lemma, $f_{n+1}$ has $n+3$ graph nodal domains. However, since $\Lambda_{n+1}$ can be at most two times degenerate by the lemma above,  Courant-upper bounds \cite{KellerSchwarz2020,Davies}, gives that $f_{n+1}$ can have at most $n+2$ graph nodal domains. This is a contradiction, hence $f_{n+1}$ has to be symmetric.
\end{proof}

After this preparation, we are now in the position to prove the eigenvalue asymptotics for the modified harmonic oscillator $\tilde H_\kappa$.

\begin{proposition}\label{prop:lower_bound_modified}
    For all  $n\in \N$,  the eigenvalues $\tilde E_{n}(\kappa)$ of the operator $\tilde H_{\kappa}=\Delta + \tilde v_{\kappa}$ satisfy
    $$\liminf_{\kappa\to 0}\frac{\tilde E_{n}(\kappa)}{\kappa^2} \geq   2n+1.$$
\end{proposition}
\begin{proof}
    We compare $\tilde H_{\kappa}$ with yet another operator
    $$\hat H_{\kappa,n} = \tilde H_\kappa +\sum_{j=-k}^k (\tilde E_{n,\min} (\kappa)- \tilde E_{j,n}(\kappa))1_{I_j}, $$
    where $1_A$ is the indicator function of a set $A$ and
    \begin{align*}
        \tilde E_{n,\min}(\kappa) = \min_{j=-k,\ldots,k} \tilde E_{j,n}(\kappa).
    \end{align*}
    Next, we show that the $n$-th eigenvalue $\hat E_n(\kappa)$ of $\hat H_{\kappa,n}$ is equal to $\tilde E_{n,\min}(\kappa)$. 
    To this end, we construct a corresponding eigenfunction. Let $\tilde f_{j,n}$ be the ground states of $\tilde H_{j,\kappa}$ on the interval $I_j$, $j=-k,\ldots,k$ with norm $1/(n+1)$ which we extend by zero to $\Z$. We define the function
    $$f_n = \sum_{j=-k}^k c_{j,n}\tilde f_{j,n},$$
   where the coefficients $c_{j,n}$ are chosen such that for $j=0,\ldots,k-1$
   $$
   c_{j,n}\tilde f_{j,n}(b_j) + c_{j+1,n}\tilde f_{j+1,n}(a_{j+1})=0$$
   and $c_{-j,n}=c_{j,n}$. 
   This choice guarantees that $\hat H_{\kappa,n} f_n=0=\tilde E_{n,\min}(\kappa) f_n$ on the points between the intervals $I_j$, $j=-k,\ldots,k$. Furthermore, on $ I_j$, $j=-k,\ldots,k$, we have by construction
   $$\hat H_{\kappa,n} f_n  = \tilde H_{j,\kappa} c_{j,n}\tilde f_{j,n}+ (\tilde E_{n,\min}(\kappa)- \tilde E_{j,n}(\kappa)) c_{j,n}\tilde f_{j,n}  = \tilde E_{n,\min}(\kappa) f_n. $$
   
   Observe that $f_n$ has exactly $n$ zeros, so by Lemma~\ref{lem: zeros of excited states}, we conclude that $f_n$ is the eigenfunction to the $n$-th eigenvalue $\tilde E_{n,\min}(\kappa)$ of $\hat H_{\kappa,n}$.
   Moreover, by Lemma~\ref{lem:lower_bound_interval} and \ref{lem:lower_bound_modified}
    \begin{align*}
      \hat E_n(\kappa)= \tilde E_{n,\min}(\kappa)=\min_{j=-k,\ldots,k} \frac{\tilde E_{0,j}(\kappa)}{\kappa^2} \geq (2n+1) +o(\kappa^2).
    \end{align*}

   Since $(\tilde E_{n,\min} (\kappa)- \tilde E_{j,n}(\kappa))\le0$ for all $j=-k,\ldots,k$, we have
   $\tilde H_{\kappa,n}\geq \hat H_{\kappa,n}$ and, thus, $\tilde E_{n} (\kappa)\ge \hat E_n(\kappa)$ by the min-max principle \cite{RS_78}.
    This concludes the proof.
\end{proof}

\subsection{IMS localization formula}
Above we have shown the eigenvalue asymptotics for the modified harmonic oscillator $\tilde H_\kappa$. To conclude the proof of Theorem~\ref{thm:eigenvalue_asymptotics_harmonic_oscillator}, we now employ the IMS localization formula to compare $H_\kappa$ and $\tilde H_\kappa$.  To this end, we first state the IMS localization formula in our discrete setting. This formula goes back to Ismagilov/Morgan/Simon \cite[Chapter~3.1]{CyconFroeseKirschSimon}. Here, we give a slightly more general version than needed  for the one dimensional case since we employ it also later for more general potentials.

As usual, we identify a scalar function on $\Z^d$ with its corresponding operator of multiplication on $\ell^2(\Z^d)$. Furthermore, for two bounded operators $A$ and $B$ on $\ell^2(\Z^d)$,
we denote its commutator by
\[[A,B]=AB-BA.\]
We call a bounded operator $L$ on $C(X)$ which is bounded on  $\ell^2(\Z)$ a \emph{Laplace type operator} if its matrix elements satisfy $L(x,y)\le 0$ for all $x\neq y$ and $L1=0$.

\begin{lemma}[IMS localization formula]\label{lem:IMS}
Let $\eta_0,\ldots,\eta_m:\Z\to [0,1]$ be such that $\sum_{j=0}^m\eta_j^2=1$. Then, for any bounded self-adjoint Laplace type operator $L$ on $\ell^2(\Z^d)$ and a potential $V$, we have for $H=L+V$
\[H= \sum_{j=1}^m \eta_j H \eta_j +\frac12\sum_{j=1}^m [\eta_j,[\eta_j,L]].\]
Furthermore, for all $j=1,\ldots,m$
\[\| [\eta_j,[  \eta_j,H]]\|\leq 2\|L\| C^2,\]
where
\[C=\sup_{|x-y|=1} |\eta_j(x)-\eta_j(y)|.\]
\end{lemma}
\begin{proof}
The equality follows by direct calculation using
\begin{align*}
[\eta_j,[\eta_j,L]]=\eta_j^2 L + L \eta_j^2 - 2\eta_j L \eta_j
\end{align*}
and summing over $j=1,\ldots,m$.

To get the norm estimate we calculate, for $f\in \ell^2(\Z^d)$,
\[\langle f, [\eta_j,[  \eta_j,H]]f\rangle = \sum_{x,y\in \Z^d}  f(x)f(y) L(x,y)(\eta_j(x)-\eta_j(y))^2 .\]
Hence, we have by the Cauchy-Schwarz inequality
\begin{align*}
    |\langle f, [\eta_j,[  \eta_j,H]]f\rangle| &\leq \sum_{x,y\in \Z^d}  |f(x)|^2 |L(x,y)|(\eta_j(x)-\eta_j(y))^2 \\
    &\leq \left(\sup_{x\in \Z^d} \sum_{y\in \Z^d} |L(x,y)|\right)C^2 \|f\|^2\leq 2\|L\|C^2\|f\|^2.
\end{align*}
This concludes the proof.
\end{proof}

This lemma provides us the tool to prove the lower bound for the eigenvalues of $H_\kappa$ by comparing it to $\tilde H_\kappa$. We follow the proof of \cite[Theorem 3.2]{Simon_1983}

\begin{proof}[Proof of Proposition~\ref{thm:lowerbound}] \lorenzorrection{The proof works by induction on the energy level index $n\geq 0$. We  argue the induction step and will see that the base case  $n=0$ is indeed a special case of the induction step (where no induction hypothesis is needed). Suppose we have proven for all $m\leq n-1$ that
\begin{equation*}
    E_m(\kappa) \geq \kappa^2(2m+1) + o(\kappa^2).
\end{equation*} We shall now prove that the statement is valid the case $m=n$.}

We construct a partition of unity. Let $ \eta :\R\to[0,1]$ be given by $\eta(y) =0\vee(2-|y|  )\wedge 1$. Then, $1_{[-1,1]}\leq \eta\leq 1_{[-2,2]}$ and $|\eta(y)-\eta(y')|\leq |y-y'|$ for all $y,y'\in \R$. We let for $\kappa>0$ small enough
   $\eta_0,\eta_1:\Z\to [0,1]$ by
   $$\eta_1(x)=\eta(2\floor{\kappa^{-(1+\delta)}}^{-1} x ),\qquad\eta_0=\sqrt{1-\eta_1^2}.$$
   Furthermore, we have by construction $\eta_1(x_\delta)=0$  at $x_\delta$ with $|x_\delta|=\floor{\kappa^{-(1+\delta)}}$. Hence, $\eta_1v_\kappa=\eta_1\tilde v_\kappa$ and therefore, $\eta_1 H_\kappa \eta_1=\eta_1 \tilde H_\kappa \eta_1$.

  Moreover, $\eta_0^2+\eta_1^2=1$ and, by the IMS localization formula, Lemma~\ref{lem:IMS}, we have for the harmonic oscillator $H_\kappa=\Delta + v_\kappa$
   \[H_\kappa= \eta_0 H_\kappa \eta_0 +\eta_1 \tilde H_\kappa \eta_1+\frac12\sum_{j=0}^1 [\eta_j,[\eta_j,\Delta]].\]
   We now estimate these terms separately. First of all, since $\Delta\ge0$ and  $\eta_0$ is supported on the set $|x|\geq \floor{\kappa^{-(1+\delta)}}$, we have for the first term 
\begin{align*}
\eta_0 H_\kappa \eta_0 \ge \eta_0 v_\kappa \eta_0 \geq  \kappa^{2(1-\delta)} \eta_0^2 \ge \kappa^2(2n+1) \eta_0^2   
\end{align*}
    for  $\kappa$ small enough. 
    
  \lorenzorrection{ Let $r$ be a number such that $2(n-1)+1 < r < 2n+1$ and denote by $P_{N,n-1}$ the projector on the  eigenfunctions of $\widetilde{H}_\kappa$ for eigenvalues less than $\kappa^{2} r$ and $P_{N,0}=0$. By the induction hypothesis, we have that $P_{N,n-1}$ has rank $n-1$ for small enough $\kappa$ and trivially we have that $P_{N,0}$ has rank $0$ (which is the base case of the induction).
    By the eigenvalue lower bound for the modified harmonic oscillator proved in Proposition~\ref{prop:lower_bound_modified} above, we have for the second term  $\eta_1 \tilde H_\kappa \eta_1$ in the IMS localization formula that
   \begin{multline*}
   \eta_1 \tilde H_\kappa \eta_1  = \eta_1 \tilde{H}_\kappa P_{N,n-1}\eta_1 + \eta_1 \tilde{H}_\kappa(I-P_{N,n-1})\eta_1\\
   \ge \kappa^2r \eta_1^2  + \eta_1P_{N,n-1}(\tilde H_\kappa - \kappa^2r)P_{N,n-1}\eta_1+o(\kappa^2)
   \end{multline*}
   for $\kappa$ small enough. We note that $\eta_1 P_{N,n-1}(\tilde H_\kappa - \kappa^2r)P_{N,n-1}\eta_1$ has rank at most $n-1$.\\
   Finally, by the Lipschitz property of $\eta_j$, we have
   \[\| [\eta_j,[  \eta_j,\Delta]]\|\leq 2\cdot 2\cdot (2\kappa^{1+\delta})^2= 16\kappa^{2(1+\delta)}=o(\kappa^2).\]
   Thus, by the min-max principle and $\eta_0^2+\eta_1^2=1$, we have for the $n$-th eigenvalue $E_n(\kappa)$ of $H_\kappa$ 
   \[E_n(\kappa)\geq \kappa^2r+o(\kappa^2).\]
  which yields the desired lower bound since we can choose $r$ arbitrarily close to $2n+1$.} This concludes the proof of Proposition~\ref{thm:lowerbound}.\end{proof}
Hence, the proof of the eigenvalue asymptotics for the rescaled harmonic oscillator $H_\kappa$ follows immediately.
\begin{proof}[Proof of Theorem~\ref{thm:eigenvalue_asymptotics_harmonic_oscillator}]The upper bound follows from Proposition~\ref{thm:upperbound} and the lower bound from Proposition~\ref{thm:lowerbound}.\end{proof}

By rescaling we deduce the eigenvalue asymptotics for the operator $H_N$   from the one for $H_\kappa$ and separation of variables. 

\begin{proof}[Proof of Theorem~\ref{thm:eigenvalue_asymptotics_harmonic_oscillator_main}]
    The higher dimensional case follows directly from the one-dimensional case by separation of variables. So, consider the operator $H_N$ for $d=1$.
    As already discussed above by setting $\kappa=\sqrt{\omega/N^{1+\gamma}}$, one sees that
$$H_N = \frac{N^2}{2}H_\kappa.$$
Hence, if $E_n(H_N)$ are the eigenvalues  of the operator $H_N$ for $d=1$, then for $\lambda_N=N^{1-\gamma}$ and $\kappa=\sqrt{\omega/N^{1+\gamma}}$, we have
$$\frac{E_n(H_N)}{\lambda_N}= \frac{N^2 E_n(\kappa)}{2 N^{1-\gamma}}= \frac{ E_n(\kappa)}{ N^{-(1+\gamma)}} =\frac{\omega}{2}\frac{E_n(\kappa)}{\kappa^2}.$$
Thus, taking the limit of the right hand side for $\kappa\to 0$, we obtain the corresponding limit for the left hand side as $N\to\infty$. Here, the choice $\gamma>-1$ ensures that $\kappa\to 0$ as $N\to\infty$. This concludes the proof.    
\end{proof}

\section{General potentials}\label{sec:general_potentials}
In this section, we extend the eigenvalue asymptotics from the harmonic oscillator potential to more general potentials satisfying Assumption~\ref{assumption}.

We first prove the upper bound on the eigenvalues by constructing suitable test functions based on the eigenfunctions of the harmonic oscillator near the minima of the general potential.
The main idea for the lower bound is to use the IMS localization formula to compare the general potential with the harmonic oscillator potential near its minimum.

\subsection{Eigenvalue upper bound for general potentials}

\begin{proposition}[Upper bounds general potentials]\label{prop:upper_bound_general_potential}
    Let $V:\R^d\to [0,\infty)$ be a potential satisfying Assumption~\ref{assumption} and let $H_N=\frac{N^2}2\Delta + \lambda_N^2 V_N$ with $\lambda_N=N^{1-\gamma}$, $\gamma\in(-1,1)$. Then, for all  $n\in \N$,  the eigenvalues $E_{n}(H_N)$ of the operator $H_{N}$ satisfy
    $$\limsup_{N\to \infty}\frac{E_{n}(H_N)}{\lambda_N} \le   e_n(V).$$ 
\end{proposition}

We first prove two lemmas which will also later be used for the lower bound. Consider a potential $V$ satisfying Assumption\ref{assumption} and let $a_l$, $l=1,\ldots,m$ be its minima. We compare the corresponding Schrödinger operator to the harmonic oscillator consicer $H_{N,l}=\frac{N^2}2\Delta+\lambda_N^2 V_{N,l}$ with $V_{N,l}(x) =V^{harm}({\omega(a_l)}(x-Na_l))=\frac{\omega(a_l)^2}2(x-Na_l)^2$ with minimum at $Na_l$ and $\lambda_N=N^{1-\gamma}$, $\gamma\in (-1,1)$. For $a\in \R^d$ and $r>0$,  we define the cube 
\begin{align*}
    Q_{r}(a)=\bigtimes_{j=1}^{d}\left[a_{j}-r,a_{j}+r\right]
\end{align*}
which is the ball with respect to the $\ell^\infty$ norm $|y|_\infty=\max\{|y_1,|\ldots,|y_d|\}$, $y\in \R^d$, centered at $a$ with radius $r$.
Clearly, for $|a-b|_\infty>2r$, the cubes $Q_r(a)$ and $Q_r(b)$ are disjoint.

\begin{lemma}[Estimating the potential error]\label{lem:potential_difference}
    Let $V:\R^d\to [0,\infty)$ be a potential satisfying Assumption~\ref{assumption}~(V1) and~(V2) and $N\in \N_0$. Then, for any normalized $f$ supported in $Q_{r}(Na)$ with $0<r<N$, we have
    \begin{align*}
       \langle f, (H_N-H_{N,l})f\rangle = O(\lambda_N^2(r/N)^3).
    \end{align*}
\end{lemma}
\begin{proof}
     To estimate the difference of the operators, we use the Taylor expansion of $V_N$ at $Na_l$ for $ x\in Q_r(Na) $ which gives according to (V1) and (V2)
    \begin{align*}
        V_N(x) = V_N^{harm}(x) + O((|x-Na_l|_\infty/N)^3)=V_N^{harm}(x) + O((r/N)^3).
    \end{align*}
Thus, the statement follows directly.
\end{proof}

\begin{proof}[Proof of Proposition~\ref{prop:upper_bound_general_potential}]
    Let $n\in\N_0$, $a_l$ be a minimum of $V$ and $n_1,\ldots,n_d\in\N$ be such that
    \begin{align*}
        e_n(V)=\frac12\sum_{j=1}^d \omega_j(a_l)^2(2n_j+1).
    \end{align*}
    where $\omega_1(a_l)^2,\ldots,\omega_d(a_l)^2$ are the eigenvalues of the Hessian of $V$ at $a_l=(a_{l,1},\ldots,a_{l,d})$.

    We define the test function, for $\kappa_j=\sqrt{\omega_j(a_l)/N^{1+\gamma}}$, $j=1,\ldots,d$, $l=1,\ldots,m$,  via
    \begin{align*}
        \tilde\psi_{n,l,N}(x) = \prod_{j=1}^d \psi_{n_j}\left( \kappa_j (x_j-Na_{l,j}) \right)    
        =\prod_{j=1}^d (h_{n_j} \phi)\left(\kappa_j (x_j-Na_{l,j})\right).       
    \end{align*}
    To apply Lemma~\ref{lem:potential_difference}, we need to cut-off the test function outside a neighborhood of $Na_l$ since we have no control of the growth of $V$ outside of a neighborhood of its minima. To this end, we pick $Q_l=Q_{r}(Na_l)$ with $r=N^{1+\delta}/\lambda_N^{{1/2}}$ for $0<\delta< (1-\gamma)/2$  (note that such a $\delta$ exists since $\gamma<1$)
    and set
   \begin{align*}
        \psi_{n,l,N}=\tilde\psi_{n,l,N}1_{Q_l}.
    \end{align*}

    To estimate the Rayleigh quotient of $\psi_{n,l,N}$, we use the estimate from Theorem~\ref{thm:eigenvalue_asymptotics_harmonic_oscillator_main} for the harmonic oscillator  $H_{N,l}$ with minimum at $Na_l$ introduced above.

    The approximate eigenfunction property of $\psi_{n,l,N}$ for the harmonic oscillator operator $H_{N,l}$ from Lemma~\ref{lem:approx_eigenfunction} generalizes directly to the higher dimensional case by considering the one-dimensional operators in each coordinate separately. However, there is an error at each boundary point of $Q_l$ which are polynomially many.  
    Since 
    $\delta>0$, the test function $\psi_{n,l,N}$ decays superexponentially at the boundary of $Q_l$ due to the presence of the Gaußian $\phi$. Hence, this error is also of order $o(\lambda_N)$ for $N$ large enough. We obtain with the proper rescaling of the approximate eigenfunctions that Lemma~\ref{lem:approx_eigenfunction}  
    \begin{align*}
        \langle \psi_{n,l,N},H_{N,l} \psi_{n,l,N}\rangle = \lambda_N e_n(V) \|\psi_{n,l,N}\|^2 + o(\lambda_N)\|\psi_{n,l,N}\|^2.
    \end{align*}
    Putting these estimates together with Lemma~\ref{lem:potential_difference}, we obtain
    \begin{align*}
        \langle \psi_{n,l,N},H_N \psi_{n,l,N}\rangle &=    \langle \psi_{n,l,N},H_{N,l} \psi_{n,l,N}\rangle +  \langle \psi_{n,l,N},(H_N -H_{N,l})\psi_{n,l,N}\rangle\\ 
     &= \lambda_N e_n(V) \|\psi_{n,l,N}\|^2 + o(\lambda_N)\|\psi_{n,l,N}\|^2,
    \end{align*}
    where the last equality follows since $\delta<(1-\gamma)/2$ implies $r=N^{{1+\delta}}/\lambda_N^{1/2}=N^{(1-\gamma)/2+\delta}= o(\lambda_N)$. 
    
    We finish the proof by establishing the almost orthogonality of the test functions $\psi_{n,l,N}$ for different $l=1,\ldots,m$ and different $n$ associated with the multi-indices $n=(n_1,\ldots,n_d)\in \N_0^d$ and then applying the min-max principle.

    Observe that the Gaußian  in $\psi_{n,l,N}$ is centered at $Na_l$ and decays superexponentially away from $Na_l$.
     Since for different $a_{l'}\neq a_l$, the distance $|Na_{l'}-Na_l|$ grows linearly in $N$ while $r=N^{1+\gamma}/\lambda_N^{1/2}<N$ as $\delta<(1-\gamma)/2$, so the supports of  $\psi_{n,l,N}$ and  $\psi_{n,l',N}$, $l\neq l'$ are eventually disjoint for $N$ large enough due to the presence of $1_{Q_l}$ in $\psi_{n,l,N}$.  Furthermore, Lemma~\ref{lem:almost_orthogonality} generalizes directly to the higher dimensional case by considering the one-dimensional operators in each coordinate separately and observe
\begin{align*}
    \langle \psi_{n,l,N},\psi_{n',l }\rangle=
    \langle \tilde\psi_{n,l,N},\tilde\psi_{n',l }\rangle+O(1)
\end{align*}
since
\begin{align*}
    \sum_{|x|\ge \kappa^{-1-\delta}} |\tilde \psi_{n,l,N}(x)|^2 &\leq C\frac{e^{-c\lambda_N^\delta}}{\lambda_N}\sum_{|x|\ge \kappa^{-1-\delta}}  \left|x\right|^{n+n'}e^{-c\lambda_N^{-1-\gamma}x^2}\le C. 
\end{align*}
Thus, for $l,l'=1,\ldots,m$ and $n,n'\in \N_0$, there exists a constant $C_{n,n'}>0$ such that for $N$ large enough
    \begin{align*}
    \left|\langle \tilde\psi_{n,l,N},\tilde\psi_{n',l }\rangle- \frac{ \sqrt{\pi}2^nn!}{\kappa }\delta_{(n,l),(n',l')} \right| \leq C_{n,n'} .
\end{align*}
Therefore, the result follows by the min-max principle as in the proof of Proposition~\ref{thm:upperbound}.
\end{proof}

\subsection{Eigenvalue lower bound for general potentials}
Next, we prove the corresponding lower bound on the eigenvalues for general potentials using the IMS localization formula. Indeed, the proof is similar to the one of Proposition~\ref{thm:lowerbound} for the harmonic oscillator and is consequently inspired by the proof of \cite[Theorem 3.2]{Simon_1983}.

\begin{proposition}\label{prop:lower_bound_general_potentials}
    Let $V:\R^d\to [0,\infty)$ be a potential satisfying Assumption~\ref{assumption} and let $H_N=\frac{N^2}2\Delta + \lambda_N^2 V_N$ with $\lambda_N=N^{1-\gamma}$, $\gamma\in(-1,1)$. Then, for all  $n\in \N$,  the eigenvalues $E_{n}(H_N)$ of the operator $H_{N}$ satisfy
    $$\liminf_{N\to \infty}\frac{E_{n}(H_N)}{\lambda_N} \ge   e_n(V).$$ 
\end{proposition}
\begin{proof}
    \lorenzorrection{ As in the proof of  Proposition~\ref{thm:lowerbound}, the proof works again by induction on the energy level index $n\geq 0$. As above the base case $n=0$ can be obtained as a special case of the proof below without the inductive hypothesis. Suppose we have proved for all $m\leq n-1$ that\begin{equation*}
    E_m(H_N) \geq \lambda_Ne_{m}(V) + o(\lambda_N).
\end{equation*} We shall now prove that the statement holds true for $m=n$.} 

We employ the IMS localization formula from Lemma~\ref{lem:IMS}. 
    To this end, we choose 
      $ \eta :\R\to[0,1]$ be the given as $\eta(y) =0\vee(2-|y|_\infty )\wedge 1$ similar to the proof of Proposition~\ref{thm:lowerbound}. We then define for $l=1,\ldots,m$ and $0<\delta<(1-\gamma)/2$
    \begin{align*}
        \eta_l(x) = \eta\left(\frac{2\lambda_N^{1/2}}{N^{1+\delta}} |x - Na_l|_\infty \right)\quad\text{and}\quad
        \eta_0(x) = \sqrt{1-\sum_{l=1}^m \eta_l(x)^2}.
    \end{align*}
     Thus, $\eta_l$ is supported in the cube $Q_{r}(a_l)$ and equal to $1$ on the smaller cube $Q_{r/2}(a_l)$ with radius $r={\lambda_N^{1/2}}/{N^{1+\delta}} $.    
    Hence, the upper bound  on $\delta$ ensures that $\eta_k$ and $\eta_l$ for $k\neq l$ have eventually disjoint support for $N\to\infty$.
    Furthermore, we estimate the Lipschitz constant $C$ from Lemma~\ref{lem:IMS} for this choice of partition of unity. To this end, we observe that for $|x-y|=1$, we have
    \begin{align*}
        |\eta_l(x)-\eta_l(y)| %&\leq \left|\eta\left(\frac{\lambda_N^{1/2}}{N^{1+\delta}} |x - Na_l|_\infty \right)-\eta\left(\frac{\lambda_N^{1/2}}{N^{1+\delta}} |y - Na_l|_\infty \right)\right| \\
        %&\leq \frac{\lambda_N^{1/2}}{N^{1+\delta}} | |x - Na_l|_\infty - |y - Na_l|_\infty | 
        \leq \frac{2\lambda_N^{1/2}}{N^{1+\delta}}.
    \end{align*}

    We now apply the IMS localization formula from Lemma~\ref{lem:IMS} with this choice  to obtain
    \begin{align*}
        H_N = \eta_0 H_N \eta_0+ \sum_{l=1}^m \eta_l H_{N,l} \eta_l +  \sum_{l=1}^m \eta_l (H_N- H_{N,l}) \eta_l+ \frac{N^2}2\sum_{l=0}^m [\eta_l,[\eta_l,\Delta]].
    \end{align*}
We start to estimate the last term. By the Lipschitz estimate above, we have
    \begin{align*}
        N^2\| [\eta_l,[  \eta_l,\Delta]]\|\leq 2\cdot 2N^2\cdot \left(\frac{2\lambda_N^{1/2}}{N^{1+\delta}}\right)^2=  16\frac{\lambda_N}{N^{2\delta}}=o(\lambda_N)
    \end{align*}
    since $\delta>0$. Secondly, we estimate the difference term using Lemma~\ref{lem:potential_difference} which gives
    \begin{align*}
       \langle f, \eta_l (H_N- H_{N,l}) \eta_l f\rangle = O(\lambda_N^2(r/N)^3) = O\left(\lambda_N^2\left(\frac{\lambda_N^{1/2}}{N^{1+\delta}N}\right)^3\right) = o(\lambda_N)
    \end{align*}
    since $\delta<(1-\gamma)/2$ and $r={\lambda_N^{1/2}}/{N^{1+\delta}} $.

   To estimate the first term, we observe that on the support of $\eta_0$, we have for $N$ large enough using $\Delta\geq0$ and the assumption (V3) that $\liminf_{x\to\infty}V(x)\ge c$
    \begin{align*}
        \eta_0 H_N \eta_0 \geq \lambda_N^2 \eta_0^2 c \geq \lambda_N e_n(V) \eta_0^2.
    \end{align*}
    for all $n$ whenever $N$ is large enough.

\lorenzorrection{Now, we estimate the second term $ \sum_{l} \eta_l H_{N,l} \eta_l $. We let a number $r$ such that $e_{n-1}(V) < r < e_{n}(V)$ be given and denote by
$P_{N,n-1,l}$  the orthogonal projection onto the span of the eigenfunctions of $H_{N,l}$  having energy less or equal than $\lambda_Nr$ and $P_{N,0,l}=0$. Thus, by induction hypothesis and the eigenvalue asymptotics of the harmonic oscillator, Theorem~\ref{thm:eigenvalue_asymptotics_harmonic_oscillator_main}, we have that the sum  $\sum_l P_{N,n-1,l}$ has rank at most $n-1$ for $N$ large enough and trivially $ P_{N,0,l} $ has rank $0$ which is needed for the base case.}
\lorenzorrection{Then,  we have
    \begin{align*}
        H_{N,l} &\geq \lambda_N r(I-P_{N,n-1,l}) + P_{N,n-1,l}H_{N,l}P_{N,n-1,l} + o(\lambda_N)\\&= \lambda_N r+ P_{N,n-1,l}(H_{N,l}-\lambda_N r)P_{N,n-1,l} + o(\lambda_N).
    \end{align*}
Now, the operator
\begin{align*}
  Q_{N,n-1}=  \sum_{l=1}^m \eta_l P_{N,n,l}(H_{N,l}-\lambda_N r)P_{N,n,l} \eta_l 
\end{align*}
has rank at most $n-1$ as discussed above.
Thus, taking all estimates together, we obtain
    \begin{align*}
        H_N \geq \lambda_N r + Q_{N,n-1} + o(\lambda_N).
    \end{align*}
    By the min-max principle and  since $r$ can be chosen arbitrarily close to $e_n(V)$, this yields the desired lower bound on the eigenvalues.}    
\end{proof}

\begin{proof}
[Proof of Theorem~\ref{thm:main}]
The upper bound follows from Proposition~\ref{prop:upper_bound_general_potential} and the lower bound from Proposition~\ref{prop:lower_bound_general_potentials}.
\end{proof}

\section{Alternative scaling regimes}\label{sec: Alternative regions}
In this section, we study the eigenvalue asymptotics of the the harmonic oscillator potential for $\gamma \in (-\infty, -1 ]$ in the one-dimensional case $d=1$ in contrast to the interval $ (1,\infty) $ which was studied in Section~\ref{sec: Discrete harmonic oscillators}. The $d>1$ cases then follow by tensorization of the former results for $d=1$. These techniques might be generalized to more arbitrary potentials. This goes beyond the scope of the present work.

Let us take $\omega >0$ and consider the multiplication operator on $\ell^2(\mathbb{Z})$ by the function $$v\colon \mathbb{Z}\to\mathbb{R},\qquad v(x) = \frac{\omega^2 x^2}{2}.$$
We denote by $$H_N = \frac{N^2}{2}\Delta + N^{-2\gamma}v$$ the discrete harmonic oscillator operator and by $E_{n}(H_N)$ its \textit{n}-th eigenvalue counted with multiplicity. 
 
 We first analyze the singular value in the parameter space $\gamma= -1$, where $\lambda_N = N^2$.  
 \begin{proposition}\label{prop: -1}
 	For  $\gamma = -1$, it holds all $n\in\mathbb{N}_0$ 
    \begin{equation*}
 		\lim_{N\to+\infty}\frac{E_n(H_N)}{N^2} = E_n(H_1) 
 	\end{equation*}
 \end{proposition} 
 \begin{proof}
 	For $\gamma= -1$ the discrete harmonic oscillator becomes \begin{equation*}
 		H_N = N^2\left(\frac{\Delta}{2} +v\right) = N^2 H_1.
 	\end{equation*}Therefore, the statement of the proposition follows immediately.
 \end{proof}
 
 Next, we analyze the parameter domain  $(-\infty, -1)$. We have the following result.
 \begin{proposition}\label{prop: -inf to -1}
 	For  $\gamma\in (-\infty, -1)$, it holds \begin{align*}
 		&\lim_{N\to+\infty} \frac{E_{0}(H_N)}{N^{2|\gamma|}}  = 0\\
 		&\lim_{N\to+\infty} \frac{E_{2n}(H_N)}{N^{2|\gamma|}} =  \lim_{N\to+\infty} \frac{E_{2n-1}(H_N)}{N^{2|\gamma|}} = \frac{\omega^2 n^2}{2},\qquad n\geq1.
 	\end{align*}
 \end{proposition} \begin{proof}
For $\gamma <-1$, the harmonic oscillator operator can be written as \begin{equation*}
	H_N = N^{2|\gamma|} \left( N^{2-2|\gamma|}\frac{\Delta}{2} + v\right).
\end{equation*} 
It is easier to study the operator $H_N/ N^{2|\gamma|}$, where the only dependence on $N$ is now in the factor $N^{2-2|\gamma|}$ front of the discrete Laplacian. Since $\Delta$ is bounded, the Laplacian term vanishes as\begin{equation*}
\lim_{N\to +\infty }N^{2-2|\gamma|}\Delta = 0,
\end{equation*} where the latter limit is in the $\ell^2(\mathbb{Z})$-norm topology. 
{
We infer by \cite[Chapter VIII, Theorem~3.6 \&~3.15]{Kato} that the eigenvalues of $H_N/N^{2|\gamma|}$ converge to the eigenvalues of $v$.}

 The multiplication operator $v$ has a purely discrete spectrum with eigenfunctions given by $$\delta_n \colon \mathbb{Z}\to\mathbb{R},\qquad \delta_n(x) = \begin{dcases}
	1 & x = n\\
	 0& x\neq n
\end{dcases},$$ for all $n\in\mathbb{Z}$. In particular, $\delta_0$ is the unique ground state with energy given by $E_0(v) = 0$, whereas the $2n$-th and $2n-1$-th excited states are degenerates, having the same energy $$E_{2n}(v) = E_{2n-1}(v) = v(n) = \frac{\omega^2 n^2}{2}. $$ 
This finishes the proof of the proposition.
 \end{proof} 
 
 \begin{remark}\label{rmk: change of scaling}
 We note that in the regime $\gamma\in (-\infty, -1)$, the scaling of $E_n(H_N)$ as a function of $\gamma$ changes in comparison to the case $\gamma\in[-1,+\infty)$. Indeed, in this latter interval we have by Theorem~\ref{thm:eigenvalue_asymptotics_harmonic_oscillator_main} and Proposition~\ref{prop: -1}
 \begin{equation*}
 	\frac{E_n(H_N)}{\lambda_N} = \frac{E_n(H_{N})}{N^{1-\gamma}} = e_n(V) + o(1),
 \end{equation*} as $N\to+\infty$, while for $\gamma\in (-\infty,-1)$, Proposition \ref{prop: -inf to -1} gives \begin{equation*}
\frac{ N^2 E_n(H_N)}{\lambda_N^2}  =  \frac{E_n(H_N)}{N^{-2\gamma }} = v(n) + o(1),
 \end{equation*}as $N\to+\infty$.
 \end{remark}

We have been able to obtain a complete description of the asymptotic behavior of the harmonic oscillator eigenvalues  for all $\gamma\in (-\infty,+\infty)$. We summarize in the following our results for the latter operator. For a fixed $n\in\mathbb{N}_0$, for all $N$ sufficiently large, the $(n+1)$-th eigenvalue of $H(N)$ satisfies \begin{equation*}
     E_n(H_N) = C_n(\gamma) N^{h(\gamma)} + o\left(N^{h(\gamma)}\right)
 \end{equation*}\begin{itemize}
     \item For $\gamma >-1$,  $C_n(\gamma) = e_n(V)$ and $h(\gamma) = 1-\gamma$
     \item For $\gamma =-1$, $ C_n(-1) = E_n(H_1)$ and $h(-1) = 2$
     \item For $\gamma <-1$ we have to discuss two different cases: \begin{enumerate}
     \item [(a)] The case $n= 0$ correspond to the ground state. We have $C_0(\gamma) = 0 $ and one can still choose $h(\gamma) = 2|\gamma|$ in the error term.
     \item [(b)] For $n\geq 1$ the even and odd levels become approximately degenerate, that is we have $C_{2n}(\gamma) = C_{2n-1}(\gamma) = \omega^2 n^2 /2$ and $h(\gamma) = 2|\gamma|$.
     \end{enumerate}
 \end{itemize} 
 
 \begin{figure}[h]
               \centering
               \begin{tikzpicture}[scale = 0.8]
                   \draw[->,thick] (7.5,.5) --(7.5,6.5);
                   \draw[->,thick] (.5,2) -- (14.5,2);
                    \node  at (11,6) {$h(\gamma) = \lim_{N\to+\infty}\frac{\ln(E_n(H_N))}{\ln(N)}$};
                    \node at (14.5,2.5) {$\gamma$};
                    \node at (5,1.5) {$-1$};
                     \node at (10,1.5) {$+1$};
                     \node at (14,1.5) {$+\infty$};
                     \node at (7.8,1.5)  {$0$};
                     \node at (1,1.5) {$-\infty$};
                    \draw[thick] (10,1.8)--(10,2.2);
                    \draw[thick] (5,1.8)--(5,2.2);
                    \draw[thick] (7.3,3)--(7.7,3);
                    \draw[thick] (7.3,4)--(7.7,4);
                    \draw [ thick] (5,4)-- (10,2);
                    \draw [ thick] (5,4)-- (3,6);
                    \draw [ thick] (10,2)-- (12.5,1);
                    \node at (8,3) {$1$};
                     \node at (8,4) {$2$};
                    \draw [fill]  (5,4) circle (.7mm);
               \end{tikzpicture}\caption{Dependence on $\gamma$}
               \label{fig: scaling with gamma}
           \end{figure}
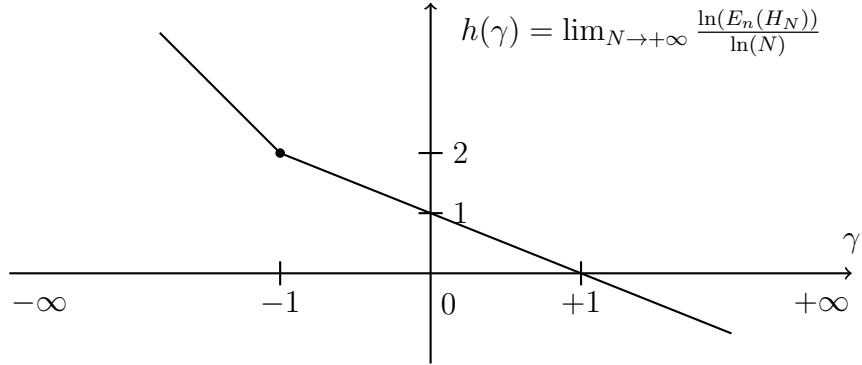

As we remarked in the remark above, the different scaling of the eigenvalues for $\gamma\in (-\infty,-1)$ and $\gamma \in [-1,+\infty)$ can be observed by the non differentiability of the exponent function $h(\gamma)$ at $\gamma = -1$. This is graphically represented in Figure \ref{fig: scaling with gamma}.

 \vspace{5mm}
\textbf{Acknowledgements.} MK acknowledges the financial support of the DFG and the hospitality of the IIAS, Jerusalem.  LP is grateful for the support of the National Group of Mathematical Physics (GNFM-INdAM). The authors thank Elke Rosenberger for most valuable hints on the literature.

\bibliographystyle{alpha}

%\printbibliography
\end{document}